\documentclass[11pt,reqno]{amsart}
\usepackage[utf8]{inputenc}
\usepackage[margin=2.5cm]{geometry}

\usepackage[nospace,noadjust]{cite}
\usepackage{relsize}
\usepackage{color}
\usepackage[dvipsnames]{xcolor}

\usepackage{tikz-cd}
\usetikzlibrary{calc,fit,matrix,arrows,automata,positioning}
\usetikzlibrary{decorations.pathreplacing,angles,quotes}
\usepackage{colortbl}
\usepackage{hyperref,enumitem}
\hypersetup{
  colorlinks   = true,
  urlcolor     = blue,
  linkcolor    = RoyalBlue,
  citecolor   = red
}
\usepackage{amsmath,amsthm,amssymb,mathtools}
\usepackage{stmaryrd}
\usepackage{array}
\newcolumntype{x}[1]{>{\centering\arraybackslash\hspace{0pt}}p{#1}}

\theoremstyle{definition}
\newtheorem{theorem}{Theorem}[section]
\newtheorem{definition}[theorem]{{{Definition}}}
\newtheorem{example}[theorem]{{{Example}}}

\newtheorem{remark}[theorem]{{{Remark}}}
\newtheorem{corollary}[theorem]{{{Corollary}}}
\newtheorem{proposition}[theorem]{{{Proposition}}}
\newtheorem{lemma}[theorem]{{{Lemma}}}

\newcommand{\numberset}{\mathbb}
\newcommand{\N}{\numberset{N}}
\newcommand{\Z}{\numberset{Z}}
\newcommand{\F}{\numberset{F}}

\newcommand{\mC}{\mathcal{C}}
\newcommand{\mP}{\mathcal{P}}
\newcommand{\mG}{\mathcal{G}}
\newcommand{\mL}{\mathcal{L}}
\newcommand{\mN}{\mathcal{N}}

\newcommand{\mU}{\mathcal{U}}
\renewcommand{\mG}{\mathcal{G}}
\newcommand{\mV}{\mathcal{V}}

\newcommand{\mS}{\mathcal{S}}

\newcommand{\GH}{G^\mathrm{H}}
\newcommand{\mCH}{\mC^\mathrm{H}}

\usepackage{bm}
\newcommand{\abf}{\mathbf{a}}
\newcommand{\bbf}{\mathbf{b}}
\newcommand{\cbf}{\mathbf{c}}

\newcommand{\nbf}{\mathbf{n}}

\newcommand{\vbf}{\mathbf{v}}
\newcommand{\wbf}{\mathbf{w}}
\newcommand{\albf}{\bm{\alpha}}
\newcommand{\bebf}{\bm{\beta}}
\newcommand{\gabf}{\bm{\gamma}}
\newcommand{\Sa}{\mathbf{S}(\albf)}

\newcommand{\Fq}{\F_q}
\newcommand{\Fm}{\F_{q^m}}
\newcommand{\Fn}{\F_q^n}
\newcommand{\Fmn}{\F_{q^m}^n}
\newcommand{\ore}{\Fm[x;\sigma]}
\DeclareMathOperator{\Gal}{Gal}

\DeclareMathOperator{\rk}{rk}
\DeclareMathOperator{\srk}{srk}
\DeclareMathOperator{\colspan}{colsp}
\DeclareMathOperator{\rowspan}{rowsp}

\DeclareMathOperator{\GL}{GL}

\DeclareMathOperator{\id}{id}
\DeclareMathOperator{\ev}{ev}

\DeclareMathOperator{\Gab}{Gab}

\DeclareMathOperator{\LRS}{LRS}

\DeclareMathOperator{\HF}{HF}
\DeclareMathOperator{\HP}{HP}
\DeclareMathOperator{\HR}{H_{reg}}
\DeclareMathOperator{\CMR}{CM_{reg}}

\def\multiset#1{\ensuremath{\left\{\kern-.3em\left\{{#1}\right\}\kern-.3em\right\}}}

\newcommand{\verteq}{\rotatebox{90}{$\,=$}}
\newcommand{\st}{\,:\,}

\title{A $q$-analogue of the rational normal curve and Linearized Reed--Solomon Codes}
\usepackage[foot]{amsaddr}
\author{Valentina Astore$^{1,2}$}
\author{Martino Borello$^{3,1}$}
\author{Alain Couvreur$^{1,2}$}
\author{Flavio Salizzoni$^4$}
\address{$^1$Inria, France.}
\address{$^2$LIX, École polytechnique, Institut Polytechnique de Paris, France.}
\address{$^3$Universit\'e Paris 8, Laboratoire de G\'eom\'etrie, Analyse et Applications, LAGA, Universit\'e Sorbonne Paris Nord, CNRS, UMR 7539, France.}
\address{$^4$MPI Leipzig}

\email{valentina.astore@inria.fr}
\email{martino.borello@univ-paris8.fr}
\email{alain.couvreur@inria.fr}
\email{flavio.salizzoni@mis.mpg.de}
\thanks{The first three authors are partially supported by the ANR-21-CE39-0009 - BARRACUDA (French \emph{Agence Nationale de la Recherche}).}
\thanks{V.~A. is funded by AID (French \emph{Agence de l'innovation de défense}).}

\begin{document}

\begin{abstract} 
    The relationship between linear codes in the Hamming metric and projective algebraic varieties has led to deep interactions between coding theory and algebraic geometry, with classical examples such as Reed--Solomon codes and the rational normal curve. 
    On the other hand, the sum-rank metric has recently gained attention due to applications in network coding, distributed storage, and post-quantum cryptography, with linearized Reed--Solomon codes emerging as optimal constructions. Despite recent advances, their structural and geometric properties are still not fully understood, and existing distinguishers remain limited. In this paper, we develop a geometric framework for linearized Reed--Solomon codes by considering a $q$-analogue of the rational normal curve.
    This yields a geometric characterization for certain parameter choices and reveals that the corresponding sets of points satisfy unexpectedly many $(q+1)$-degree hypersurface conditions.

Our approach extends Schur-product-based techniques from the Hamming and rank-metric settings to the sum-rank metric case. Finally, we study the Hilbert function of the associated coordinate ring, providing a detailed description of its behavior and identifying its regularity, which also sheds new light on Gabidulin codes.
\end{abstract}
\maketitle


\tableofcontents

\section*{Introduction} 

The connection between linear codes endowed with the Hamming metric and projective algebraic varieties is well known and has proved to be highly fruitful. It links problems that may appear to belong to quite different areas: for example the MDS conjecture and the existence of large arcs \cite{segre1955curve,ball2012sets}, minimal codes and strong blocking sets \cite{alfarano2022geometric}, covering problems and saturating sets \cite{davydov2011linear}. The basic idea is that equivalence classes of linear codes can be associated with projective equivalence classes of point sets in projective space, and many coding-theoretic properties can be interpreted geometrically. A particularly significant example is given by Reed--Solomon codes, which correspond to sets of points lying on the so-called rational normal curve. An important feature of this identification is the following: it is well known that the rational normal curve is an intersection of quadrics. As a consequence, the points defined by a Reed--Solomon code lie on many more quadrics than a random set of points would. 
This property can be reinterpreted in terms of the Schur square of the code and is by now well understood \cite{mirandola2015critical}. In this paper, we aim to extend these connections to codes in another metric, namely the \emph{sum-rank metric}, and to obtain analogous results for the optimal codes in this setting, namely the so-called \emph{linearized Reed--Solomon codes}.

Sum-rank metric codes have recently attracted increasing attention due to their wide range of applications, including space-time coding, multishot network coding, and distributed storage (see~\cite{martinez2022codes} for a complete overview of these applications). Moreover, due to the existence of families of sum-rank metric codes achieving optimal parameters, they are also emerging in cryptography, where they are being explored as potential building blocks for post-quantum code-based primitives~\cite{dissertation, nouetowa:hal-05441609}. In this setting, the goal is to design schemes that aim to resist known structural attacks on constructions based on Hamming and rank-metric codes, while keeping key sizes relatively small. Since, in code-based cryptographic systems, security often relies on the difficulty of identifying the algebraic structure of the underlying code, the {\em distinguishability problem}, \emph{i.e.}, the ability to distinguish structured codes from random ones, is emerging as a central issue. If such distinguishers exist, they may undermine security proofs and potentially lead to
practical attacks.

In the last few years, sum-rank metric codes have been extensively studied in a series of papers (see for instance ~\cite{gorla2023sum,NERI2023105703,martinez2019theory,MARTINEZPENAS2018587,martinez2020hamming,caruso:hal-03395402,BC24,berardini:hal-04577005}). In this setting, codes with optimal parameters are called {\em maximum sum-rank distance (MSRD) codes}. Among the few known constructions of MSRD codes, the family of linearized Reed--Solomon codes, introduced in~\cite{MARTINEZPENAS2018587}, is the most prominent and extensively studied so far. Recently, the equivalence problem for linearized Reed--Solomon codes has been investigated in~\cite{mannaert2026number, santonastaso2025invariants}, while~\cite{hormann2022distinguishing,hormann2026distinguishers} proposed the first attempts at an algebraic characterization of this family. However, these results are not fully general since \cite{hormann2022distinguishing} requires partial knowledge of the defining parameters of the code, while~\cite{hormann2026distinguishers} applies only to some parameters choices.

Linearized Reed--Solomon codes naturally extend both Reed--Solomon and Gabidulin codes. The former being optimal for the Hamming metric and so are the latter for the rank metric, this suggests that linearized Reed--Solomon codes may inherit structural features that characterize these families. Reed--Solomon codes have been fully characterized by means of a {\em Schur-product-based distinguisher}~\cite[Corollary 27]{mirandola2015critical} (see also~\cite{randriambololona2013asymptotically}). As we mentioned above, this is related to the geometric interpretation of Reed--Solomon codes as subsets of  rational normal curve (see for example~\cite[p. 120]{arbarello2013geometry}). In the rank-metric setting, Gabidulin codes were first algebraically characterized through \emph{Overbeck’s distinguisher}~\cite{jofc-2008-14383}. Similarly to the above, Gabidulin codes can also be described geometrically in terms of linear sets. This perspective was used in~\cite{astore2025geometric} to identify Gabidulin codes via a Schur-product-based invariant, inspired by that introduced in the Hamming-metric setting. 

In the theory of linearized Reed--Solomon codes, at least two fundamental problems still have only partial answers. Firstly, even though they generalize both Reed--Solomon and Gabidulin codes, the known distinguishers for the latter have not yet been fully adapted to the sum-rank metric setting. Secondly, although the geometry of sum-rank metric codes has been investigated in~\cite{NERI2023105703,santonastaso2023subspace}, a geometric characterization of linearized Reed--Solomon codes is still lacking.

In this paper, we aim to address these gaps by investigating linearized Reed--Solomon codes from a geometric perspective. Motivated by the connection between Reed--Solomon codes and the rational normal curve, we consider a $q$-analogue of the latter, previously studied in~\cite{donati2018generalization} in the context of rank-metric codes, and we show that it captures the structure of the geometric object associated with linearized Reed--Solomon codes. In particular, we prove that linearized Reed--Solomon codes induce a partition of the set of $\Fm$-rational points of this $q$-analogue of the rational normal curve into subsets with a remarkable combinatorial structure, namely scattered linear sets, and a strong interplay among them. Some of these properties were previously established in~\cite{santonastaso2023subspace} using different techniques based on viewing linearized Reed--Solomon codes as $(k-1)$-designs and our result contributes to a deeper understanding of these objects. This description yields a geometric characterization of linearized Reed--Solomon codes for certain parameter choices.
Similarly to the case of the rational normal curve which is an intersection of quadrics, our $q$-analogue is an intersection of degree $(q+1)$ hypersurfaces: actually, it can be described as a determinantal variety. Hence, the points defined by such linearized Reed--Solomon codes lie on many more hypersurfaces of degree $q+1$ than a random set of points would. By relating sum-rank metric codes with a suitable generalization of the Hamming-metric code associated with rank-metric codes (see~\cite{ALFARANO2022105658}), we show that the method proposed here extends the Schur-product-based invariant introduced in~\cite{astore2025geometric} for Gabidulin codes. Even though the behavior of homogeneous forms of degree $q+1$ suffices to characterize linearized Reed--Solomon codes, it is natural to extend our investigation to higher degree forms, in order to obtain a finer invariant. This leads to the study of the Hilbert function of the coordinate ring associated with the set of projective points defined by our $q$-analogue of the rational normal curve. 
In particular, we obtain a comprehensive description of the behavior of this sequence and identify its points of regularity. Finally, by specializing the results obtained in the last part of the paper, we provide a deeper understanding of the geometric structure of Gabidulin codes, thereby addressing some of the open questions left in~\cite{astore2025geometric}.

\subsection*{Outline} The paper is structured as follows. Section~\ref{Preliminaries} introduces the necessary background, including the notions of Schur powers and Hilbert function in the context of coding theory, followed by the notion of sum-rank-metric codes together with an overview of their geometric interpretation, and concluding with the definition of linearized Reed--Solomon codes. In Section~\ref{The Geometry of Linearized Reed-Solomon Codes}, we develop the geometric interpretation of linearized Reed--Solomon codes, focusing in particular on their connection with a $q$-analogue of the rational normal curve, and identify the geometric properties that characterize this family. Finally, Section~\ref{The Hilbert Sequence of a Linearized Reed-Solomon Code} is dedicated to the Hilbert function of linearized Reed--Solomon codes, and shows how our approach can be adapted to better understand the geometric interpretation of Gabidulin codes.

\section{Preliminaries} \label{Preliminaries}
\subsection{Schur Powers and Hilbert Function}
Let $q$ be a prime power, $\Fq$ be a finite field with $q$ elements, and $\overline{\F}_q$ be its algebraic closure. Throughout this paper, we only consider linear codes: a $k$-dimensional code $\mC$ over $\Fn$ is a $k$-dimensional linear subspace of $\Fq^n$ endowed with a metric. For $\mathbf{v}=(v_1,\ldots,v_n),\,\mathbf{w}=(w_1,\ldots,w_n)\in\Fn$, we denote by $\star$ the standard componentwise product $\mathbf{v}\star\mathbf{w}\coloneqq(v_1w_1,\ldots,v_nw_n)$. The {\em Schur product} $\star$ of two codes $\mC_1,\mC_2\subseteq\Fn$ is defined as \[\mC_1\star\mC_2\coloneqq\langle c_1\star c_2\st c_1,c_2\in\mC\rangle_{\Fq}.\] Then, for $i\geq1$, it is natural to recursively define the $i$-th {\em Schur power} of a code $\mC\subseteq\Fn$ as \[\mC^{(i)} \coloneqq \mC\star\mC^{(i-1)},\] where $\mC^{(0)}\coloneqq\langle(1,\dots,1)\rangle_{\Fq}$ is the code generated by the all-ones vector of length $n$. For a comprehensive discussion on the Schur product of codes, we refer to~\cite{randriambololona2015onproducts}. For the purposes of this work, we only recall that $\dim\left(C^{(i)}\right)\leq\dim\left(C^{(i+1)}\right)$ for all $i\geq0$. This property motivates the following definition.
\begin{definition}\label{definition:dimseq}
    Let $\mC\subseteq \Fn$ be a linear code. The sequence of integers \[\left\{\dim\big(\mC^{(i)}\big)\right\}_{i\geq0}\] is the \emph{dimension sequence}, or \emph{Hilbert sequence}, of $\mC$. The \emph{Castelnuovo--Mumford regularity} of $\mC$ is the smallest positive integer $r=r(\mC)$ such that, for all $i\geq0$,\[\dim\big(\mC^{(r)}\big)=\dim\big(\mC^{(r+i)}\big).\]
\end{definition}
The Hilbert sequence and the Castelnuovo--Mumford regularity are closely related to the corresponding notions in commutative algebra from which they derive their names. Let $\mC$ be a linear code of dimension $k$ in $\Fq^n$ and $G\in\Fq^{k\times n}$ be a generator matrix for $\mC$. Suppose that $\mC$ has \emph{full support}, \emph{\emph{i.e.}} $G$ has no zero column. For $i\in\{1,\dots,n\}$, let $[\mathbf{g}_i]$ be the point in the projective space $\mathbb{P}^{k-1}(\overline{\F}_q)$ that corresponds to the $i$-th column $\mathbf{g}_i$ of $G$. Then, we define the {\em set of projective points associated} with $\mC$ as
\begin{equation*}
    \Pi_G\coloneqq\left\{[\mathbf{g}_i]:1\leq i\leq n\right\}\subseteq\mathbb{P}^{k-1}(\overline{\F}_q)\,.
\end{equation*}
Notice that $\Pi_G$ is a set, and hence its points are considered without multiplicity. In particular, $\lvert\Pi_G\rvert=n$ if and only if the columns of $G$ are pairwise linearly independent. While the set $\Pi_G$ depends on the chosen matrix $G$, this dependence is only up to projective equivalence, since changing the generator matrix induces a linear projective transformation of $\Pi_G$. As a result, the geometric properties of $\Pi_G$ are independent of the specific choice of $G$.
\begin{remark}\label{remark:projsyst}
    The set $\Pi_G$ lives in the projective space defined over the algebraic closure of $\F_q$. This decision was made to avoid several technical complications. However, since all the entries of $G$ are in $\F_q$, and most of the properties that we are interested in are invariant under ground field extension, one could equally work in the projective space  $\mathbb{P}^{k-1}(\F_q)$. In that case, $\Pi_G$ coincides with the underlying set of the projective system associated with the code (see~\cite{tsfasman2013algebraic} for a definition of projective system).
\end{remark}
Given a variety $\mV\subseteq\mathbb{P}^{k-1}(\Fq)$, let $I(\mV)\subseteq\Fq[x_1,\dots,x_k]$ be its vanishing ideal (\emph{i.e.}, the homogeneous ideal of polynomials that evaluate to zero on $\mV$) and $S_\mV\coloneqq\Fq[x_1,\dots,x_k]/I(\mV)$ be its {\em homogeneous coordinate ring}. We denote by $(S_\mV)_d$ the homogeneous part of degree $d$ of $S_\mV$. For $s\geq0$, \[\HF_{S_\mV}\st s \longmapsto \dim\left(\Fq[x_1,\dots,x_k]/I(\mV)\right)_s\] is the {\em Hilbert function} of $S_\mV$. It is well known that there exists a unique polynomial $\HP_{S_\mV}$, called the {\em Hilbert polynomial} of $S_\mV$, whose evaluation sequence at non-negative integers agrees with the Hilbert function for all sufficiently large $s$. The smallest integer $\bar s$ for which $\HF_{S_\mV}(s)=\HP_{S_\mV}(s)$, for all $s \geq\bar s$, is called the {\em Hilbert regularity} $\HR(S_\mV)$ of $S_\mV$. Note, in particular, that the Castelnuovo--Mumford regularity (as it is defined in the context of commutative algebra) and the Hilbert regularity of $S_\mV$ coincide if the ideal $I(\mV)$ is 0-dimensional. We refer to~\cite[Section 1.9]{eisenbud2013commutative} for further details. The following proposition illustrates the aforementioned relationship between Hilbert functions and dimension sequences.
\begin{proposition}[{\cite[Proposition 1.28]{randriambololona2015onproducts}}]\label{prop: H seq and CM reg}
    Let $\mC\subseteq\Fn$ be a linear code. Let also $G$ be a generator matrix of $\mC$ and $\Pi_G\subseteq\mathbb{P}^{k-1}(\overline{\F}_q)$ be a set of projective points associated with $\mC$. Then,
    \begin{enumerate}
        \item the Hilbert sequence of $\mC$ is equal to the Hilbert function of $\Pi_G$,
        \item the Castelnuovo--Mumford regularity $r(\mC)$ of $\mC$ coincides with the Castelnuovo--Mumford regularity of $\Pi_G$,
        \item $\dim\left(\mC^{(i)}\right)=\lvert \Pi_G\rvert$, for all $i\geq r(\mC)$,
    \end{enumerate}
    where the Hilbert function and the Castelnuovo--Mumford regularity of $\Pi_G$ are those of its homogeneous coordinate ring.
\end{proposition}

\subsection{Sum-Rank-Metric Codes and their Geometric Interpretation}
Let $t,m,n_1,\dots,n_t$ be positive integers such that $k\leq n\coloneqq\sum_{i=1}^t n_i$. We say that $n_1,\dots,n_t$ define a {\em partition} of the vector $\nbf\coloneqq(n_1,\dots,n_t)$ and we set \[\Fm^{\nbf}\coloneqq\bigoplus_{i=1}^t \Fm^{n_i}.\]

\begin{definition}
    Let $w_{\rk}(\wbf)$ denote the ($\Fq$-)\emph{rank weight} of the vector $\wbf=(w_1,\dots,w_n)\in\Fmn$, defined as $w_{\rk}(\wbf)\coloneqq\dim_{\Fq}\langle w_1,\dots,w_n\rangle_{\Fq}$. Then, the {\em sum-rank weight}, $w_{\srk}$, of $\vbf=(\mathbf{v}_1,\dots,\mathbf{v}_t)\in\Fm^\nbf$ is \[w_{\srk}(\vbf)\coloneqq\sum_{i=1}^t w_{\rk}(\mathbf{v}_i).\]
\end{definition}
It immediately follows that the sum-rank weight reduces to the rank weight over $\Fmn$ if $t=1$, while it coincides with the Hamming weight over $\Fmn$ if $n_1=\dots=n_t=1$. Furthermore, by evaluating the distance between two vectors as the weight of their difference, both the rank and the sum-rank weight naturally induce a metric over $\Fmn$ and $\Fm^{\nbf}$, respectively.
\begin{definition}
    An {\em $[\nbf,k,d]_{q^m/q}$ sum-rank-metric code} $\mC$ is a $k$-dimensional $\Fm$-linear subspace of $\Fm^{\nbf}$, endowed with the sum-rank metric. The parameter $d$ denotes the {\em minimum (sum-rank) distance} of $\mC$ and is defined as $d\coloneqq\min\{w_{\srk}(\cbf)\st\cbf\in\mC\setminus\{\mathbf{0}\}\}$. We will simply refer to an $[\nbf,k,d]_{q^m/q}$ code as an $[\nbf,k]_{q^m/q}$ code if its minimum distance is unknown or not relevant.
\end{definition}
In what follows, we restrict ourselves to reporting only what is strictly relevant for our purposes. For further information on sum-rank-metric codes, we refer the reader to~\cite{gorla2023sum,martinez2022codes}. To keep track of the original partition defined by $\nbf$, we will often represent a generator matrix $G$ of $\mC\subseteq\Fm^{\nbf}$ as the block matrix $(G_1|\cdots|G_t)$, where $G_i\in\Fm^{k\times n_i}$ for all $i\in\{1,\dots,t\}$. We say that $\mC$ is {\em non-degenerate} if the columns of each submatrix $G_i$ are linearly independent over $\Fq$. In the following, we will always assume codes to be non-degenerate. Two $[\nbf,k,d]_{q^m/q}$ codes are said to be {\em{equivalent}} if there exists a self-isometry of the ambient space that maps one onto the other. The following theorem characterizes the possible isometries of $\F_{q^m}^{\mathbf{n}}$ with respect to the sum-rank metric.
\begin{theorem}[{\cite[Theorem V.2]{camps2022optimal}}]\label{theorem:characterizationisometries}
    Let $\varphi:\F_{q^m}^{\mathbf{n}}\rightarrow\F_{q^m}^{\mathbf{n}}$ be an $\F_{q^m}$-linear isometry for the sum-rank metric. Then, there exist a permutation $\sigma:\{1,\dots,t\}\rightarrow\{1,\dots,t\}$ with the property that $\sigma(i)=j$ implies $n_i=n_j$, $a_1,\dots,a_t\in\F_{q^m}^*$, and invertible matrices $M_1\in\GL_{n_1}(\Fq),\dots,M_t\in\GL_{n_t}(\Fq)$ such that
    $$\varphi(\mathbf{v}_1,\dots,\mathbf{v}_t)=\left(a_1\mathbf{v}_{\sigma^{-1}(1)}M_1,\dots,a_t\mathbf{v}_{\sigma^{-1}(t)}M_t\right).$$
\end{theorem}
Note that, since the multiplication on the right by invertible matrices over $\F_q$ is a sum-rank-metric isometry, two  isometric $\Fm$--linear sum-rank metric codes $\mC, \mC'$ with respective generator matrices $G,G'$ might be associated with sets of points $\Pi_G, \Pi_{G'}$ that are not projectively equivalent. In light of Proposition~\ref{prop: H seq and CM reg}, this implies that the dimension sequence of codes is not invariant under sum-rank equivalence and therefore does not provide meaningful structural information about the underlying code. As a consequence, in the sum-rank-metric setting, the above defined set $\Pi_G$ is not the most suitable geometric object that a code can be associated with. The same problem already holds for rank-metric codes, and it was overcome in \cite{ALFARANO2022105658,randrianarisoa2020geometric} by defining a different geometric interpretation of codes that classifies them in terms of their equivalence classes. In a similar way, in \cite{NERI2023105703} the authors associate sum-rank-metric codes with an analogous geometric object whose properties are deeply connected with structural properties of the associated family of codes and can therefore be studied to reveal them. In the following, we work with a similar, but slightly modified, version of the geometric object introduced in \cite{NERI2023105703}. \\
Let $\mC$ be an $[\nbf,k]_{q^m/q}$ code with generator matrix $G=(G_1|\dots|G_t)$. For $i\in\{1,\dots,t\}$, consider the column space $\mU_i\coloneqq\colspan_{\Fq}(G_i)$ of $G_i$. Similarly to the classical case described in the previous subsection, we aim to derive an associated projective variety. For this reason, for every $i\in\{1,\dots,t\}$, we consider the projective set
\begin{equation*}
    \mL_{\mU_i}\coloneqq\left\{[\mathbf{u}]:\mathbf{u}\in\mU_i\setminus\{\mathbf{0}\}\right\}\subseteq\mathbb{P}^{k-1}(\overline{\F}_{q^m}).
\end{equation*}
Finally, we define the {\em set of points associated with $\mC$ with respect to the generator matrix $G$} as
\begin{equation*}
    \mL_{G}\coloneqq\mL_{\mU_1}\cup\cdots\cup \mL_{\mU_t}\subseteq\mathbb{P}^{k-1}(\overline{\F}_{q^m}).
\end{equation*}
We point out that we are always working set theoretically, and therefore up to multiplicity, as in the classical case. Regarding the cardinality of $\mL_G$, we immediately obtain the following inequality 
\begin{equation*}
    \lvert \mL_G\rvert\leq\sum_{i=1}^t\frac{q^{n_i}-1}{q-1}.
\end{equation*}
\begin{remark}
    Similarly to what is pointed out in Remark~\ref{remark:projsyst}, the set $\mL_G$ could also be considered as a set of points in the projective geometry $\mathbb{P}^{k-1}(\F_{q^m})$. In that case, $\mL_G$ coincides with the underlying set of the linear set associated with $G$. We refer to~\cite{NERI2023105703} for the definition of linear set. 
\end{remark}
The next proposition follows immediately from the definition of $\mL_G$ and Theorem~\ref{theorem:characterizationisometries}.
\begin{proposition}\label{prop: equiv codes equiv proj sets}
    Let $\mC_1$ and $\mC_2$ be two $[\nbf,k,d]_{q^m/q}$ equivalent codes with generator matrices $G_1$ and $G_2$, respectively. Then, $\mL_{G_1}$ and $\mL_{G_2}$ are equal up to a linear projective transformation.
\end{proposition}

\begin{example} \label{example: geometry of sum rank code}
    Let $q=2$, $m=3$ and let $\F_8=\F_2(\alpha)$, with $\alpha^3+\alpha+1=0$. Consider also $\nbf=(3,2)$ and $k=2$. Let $\mC$ be the $[(3,2),2]_{8/2}$ code generated by
    \[G=(G_1|G_2)=\left(\begin{array}{ccc|cc}
        1 & 0 & \alpha\, & \alpha+1 & \alpha^2 \\
        0 & 1 & 0 & 0 & 1
    \end{array}\right).\]
    Then,
    \begin{align*}
        \mU_1&=\left\{\begin{pmatrix}0 \\ 0\end{pmatrix},\begin{pmatrix}1 \\ 0\end{pmatrix}, \begin{pmatrix}0 \\ 1\end{pmatrix}, \begin{pmatrix}\alpha \\ 0\end{pmatrix}, \begin{pmatrix}1 \\ 1\end{pmatrix}, \begin{pmatrix}\alpha+1 \\ 0\end{pmatrix}, \begin{pmatrix}\alpha \\ 1\end{pmatrix}, \begin{pmatrix}\alpha+1 \\ 1\end{pmatrix}\right\} \\
        \mU_2&=\left\{\begin{pmatrix}0 \\ 0\end{pmatrix},\begin{pmatrix}\alpha+1 \\ 0\end{pmatrix}, \begin{pmatrix}\alpha^2 \\ 1\end{pmatrix}, \begin{pmatrix}\alpha^2+\alpha+1 \\ 1\end{pmatrix}\right\},
    \end{align*} and the projective sets in $\mathbb{P}^1(\F_8)$ associated with the two blocks of $\mC$ are
    \begin{align*}
        \mL_{\mU_1}&=\{[1:0], [0:1], [1:1], [\alpha:1], [\alpha+1:1]\} \\
        \mL_{\mU_2}&=\{[\alpha+1:0], [\alpha^2:1], [\alpha^2+\alpha+1:1]\}.
    \end{align*} In conclusion, the set of points associated with $\mC$ with respect to $G$ is
    \[\mL_G=\{[1:0], [0:1], [1:1], [\alpha:1], [\alpha+1:1], [\alpha^2:1], [\alpha^2+\alpha+1:1]\}.\]
\end{example}

In order to fully extend the geometric interpretation of codes introduced in the previous subsection to the sum-rank-metric setting, we conclude this paragraph by introducing the following notion (see \cite{NERI2023105703, ALFARANO2022105658}) which is designed so that its associated set of projective points faithfully captures the geometric object of a sum-rank-metric code.
Let $\sim_{\Fq}$ be the proportionality relation over $\Fm^k$ such that for all $u,u'\in\Fm^k$, $u\sim_{\Fq}u'$ if and only if $u=\lambda u'$, for some $\lambda\in\Fq^*$.
\begin{definition}
    Let $\mC$ be an $[\nbf,k]_{q^m/q}$ code with generator matrix $G=(G_1|\cdots|G_t)$. For all $i\in\{1,\dots,t\}$, let $L_i\coloneqq\lvert\mL_{G_i}\rvert$. Define $\GH\coloneqq (\GH_1\,|\cdots|\,\GH_t)$, where $\GH_i\in\Fm^{k\times L_i}$ is a matrix whose columns form a set of representatives of $(\mU_i\setminus\{\mathbf{0}\})/\sim_{\Fq}$. Then, $\mCH\coloneqq\rowspan_{\Fm}(\GH)$ is the (projective) {\em Hamming-metric code associated} with $\mC$.
\end{definition}
In particular, note that $\mCH$ is a $k$-dimensional code over $\Fm^L$, with $L\coloneqq L_1+\dots+L_t$, and, as its name suggests, it is considered to be endowed with the Hamming metric. Moreover, it was proved in \cite[Section 5.1]{NERI2023105703} that equivalent codes with respect to the sum-rank metric are mapped into equivalent associated Hamming-metric codes, with respect to the Hamming metric (to be precise, our definition is the projective version of the associated code, without repeated columns, but the result in \cite{NERI2023105703} implies the equivalence in our case).  

Consider the projective system $\Pi_{\GH}\subseteq\mathbb{P}^{k-1}(\Fm)$ associated with $\mCH$, as it was described in Remark~\ref{remark:projsyst}. Then, the relevance of the associated Hamming-metric code is immediately evident when realizing that $\Pi_{\GH}$ coincides precisely with the set of projective points $\mL_G$ associated with its underlying sum-rank-metric code $\mC=\rowspan_{\Fm}(G)$. Therefore, in view of Proposition~\ref{prop: H seq and CM reg}, in the sum-rank-metric setting it is natural to introduce the following definition, whose notation aligns with that adopted in \cite{astore2025geometric} in the context of rank-metric codes.
\begin{definition}
    The {\em $(q^m/q)$-dimension sequence}, or \emph{$(q^m/q)$-Hilbert sequence}, of a sum-rank-metric code $\mC$ is the Hilbert sequence of its associated Hamming-metric code $\mCH$.
\end{definition}
The choice of introducing this definition is further supported by the fact that equivalent $[\nbf,k]_{q^m/q}$ codes are associated with projectively equivalent sets of points in $\mathbb{P}^{k-1}(\Fm)$, as stated in Proposition~\ref{prop: equiv codes equiv proj sets}. As a result, the structural properties of a code $\mC=\rowspan_{\Fm}(G)$ which are reflected in its associated set of projective points can be studied through its corresponding $(q^m/q)$-dimension sequence. In particular, the definition of $(q^m/q)$-Hilbert sequence and Proposition~\ref{prop: H seq and CM reg} immediately imply that
\begin{align} \label{eq: Fq-Hilbert seq}
    \dim{\left(\mC^{\textrm{H}}\right)}^{(i)}&=\dim\big(\Fm[x_1,\dots,x_k]/I(\mL_G)\big)_i \nonumber \\
    &= \binom{k+i-1}{i}-\dim({I(\mL_G)}_i).
\end{align}

\begin{example}
    Consider again the code $\mC$ of Example~\ref{example: geometry of sum rank code}. Then,
    \begin{align*}
        (\mU_1\setminus\{\mathbf{0}\})/_{\sim_{\F_2}}&=\left\{\begin{pmatrix}1 \\ 0\end{pmatrix}, \begin{pmatrix}0 \\ 1\end{pmatrix}, \begin{pmatrix}\alpha \\ 0\end{pmatrix}, \begin{pmatrix}1 \\ 1\end{pmatrix}, \begin{pmatrix}\alpha+1 \\ 0\end{pmatrix}, \begin{pmatrix}\alpha \\ 1\end{pmatrix}, \begin{pmatrix}\alpha+1 \\ 1\end{pmatrix}\right\} \\
        (\mU_2\setminus\{\mathbf{0}\})/_{\sim_{\F_2}}&=\left\{\begin{pmatrix}\alpha+1 \\ 0\end{pmatrix}, \begin{pmatrix}\alpha^2 \\ 1\end{pmatrix}, \begin{pmatrix}\alpha^2+\alpha+1 \\ 1\end{pmatrix}\right\},
    \end{align*} and we obtain that the associated Hamming-metric code with $\mC$ is generated by the matrix
    \[\GH=\left(\GH_1\,|\,\GH_2\right)=\left(\begin{array}{ccccccc|ccc}
        1 & 0 & \alpha & 1 & \alpha+1 & \alpha & \alpha+1 & \alpha+1 & \alpha^2 & \alpha^2+\alpha+1 \\ 0 & 1 & 0 & 1 & 0 & 1 & 1 & 0 & 1 & 1
                                             \end{array}\right).\] It is then immediate to see that the set of projective points \[\Pi_{\GH}=\{[1:0], [0:1], [1:1], [\alpha:1], [\alpha+1:1], [\alpha^2:1], [\alpha^2+\alpha+1:1]\},\] defined by the projectivization of the columns of $\GH$, coincides with the set $\mL_G$ derived above. Moreover, the vanishing ideal over $\F_8[x_1,x_2]$ of these sets is \[I(\mL_G)=\big(x_2x_1(x_1+x_2)(x_1+\alpha x_2)(x_1+(\alpha+1)x_2)(x_1+\alpha^2x_2)(x_1+(\alpha^2+\alpha+1)x_2)\big),\] and, for all $i\geq1$, we can compute the $i$-th term of the $(2^3/2)$-Hilbert sequence of $\mC$ as
                                         \[\dim {\left(\mC^{\textrm{H}}\right)}^{(i)}=\binom{2+i-1}{i}-\dim(I(\mL_G)_i)=i+1-\dim({I(\mL_G)}_i).\]
\end{example}

\subsection{Linearized Reed--Solomon Codes}
As in the cases of linear block codes and rank-metric codes, a Singleton-like bound also holds in the sum-rank-metric setting.
\begin{theorem}[Singleton-like bound]
    Let $\mC$ be an $[\nbf,k,d]_{q^m/q}$ code. Then, \[d\leq n-k+1.\]
  \end{theorem}
  \begin{proof}
    See  \cite[Prop.~16]{MARTINEZPENAS2018587}.
  \end{proof}
  
Codes attaining the Singleton-like bound are called {\em maximum sum-rank distance separable} codes, or shortly MSRD codes. In this section, we focus on a subfamily of MSRD codes known as {\em linearized Reed--Solomon codes}. These codes were introduced by Martínez-Peñas in~\cite{MARTINEZPENAS2018587} and can be viewed as a special class of evaluation codes. To define them, we first recall the notion of skew polynomials, originally introduced in a more general setting by Ore in~\cite{ore1933theory}. \\ Let $\sigma$ be a generator of the Galois group $\Gal(\Fm/\Fq)$. The {\em skew polynomial ring} $\ore$ is the set of formal polynomials \[\bigg\{f=\sum_i f_ix^i\st f_i\in\Fm\text{ and finitely many }f_i\text{'s are different from }0 \bigg\},\] equipped with the standard componentwise polynomial addition \[f+g=\sum_i(f_i+g_i)x^i\] and the multiplication rule \[x\alpha=\sigma(\alpha)x,\] which is extended to polynomials by associativity and distributivity. Note that $\ore$ trivially reduces to the standard polynomial ring $\Fm[x]$ if $\sigma$ is the identity map, $\id_{\Fm}$, over $\Fm$. \\ The {\em degree} of a non-zero skew polynomial $f=\sum_i f_ix^i\in\ore$ is defined as $\max\{i\st f_i\neq0\}$, while we set $-\infty$ to be the degree of the zero polynomial. We will denote by $\ore_{<k}$ the set of skew polynomials of degree less than $k$.

Standard polynomial evaluation is not meaningful when considering Ore polynomials, as it does not define a ring homomorphism. Hence, it is needed to define an alternative evaluation map for the ring $\ore$. Here, we adopt the approach commonly referred to as {\em generalized operator evaluation}. For $a,b\in\Fm$, we define the evaluation map \[\ev_{a,b}:\ore\longrightarrow\Fm, \ \sum_i f_ix^i\longmapsto \sum_i f_i\left(\sigma^i(b)N_i(a)\right),\] where the $N_i$'s denote the \emph{$i$-th truncated norm} operators, namely,
\begin{equation}\label{eq:partial_norms}
  \forall a \in \Fm,\quad  N_0(a)\coloneqq 1 \qquad  \text{and}\qquad \forall i >0,\ \  N_i(a)\coloneqq\prod_{j=0}^{i-1}\sigma^j(a).
\end{equation}
In particular, $N_m$ is nothing but the usual norm on $\Fm/\Fq$ operator:
\begin{equation}\label{eq:norm}
  N_m : \Fm \longrightarrow \Fq,\ \ a \longmapsto \prod_{i=0}^{m-1} \sigma^i(a).
\end{equation}
To introduce linearized Reed--Solomon codes, we use a multi-point evaluation version of the map above. Let $\mathbf{a}=(a_1,\dots,a_t)\in\Fm^t$ and $\mathbf{b}=(\mathbf{b}_1,\dots,\mathbf{b}_t)\in\Fm^{\mathbf{n}}$, with $\mathbf{b}_i=(b_{i,1},\dots,b_{i,n_i})\in\Fm^{n_i}$ for $i\in\{1,\dots,t\}$. Then, we define
\begin{align*}
    \ev_{\abf,\bbf} :\ \ore &\longrightarrow \Fm^\nbf \\ f &\longmapsto \left(\ev_{a_1,b_{1,1}}(f),\dots,\ev_{a_1,b_{1,n_1}}(f)\,\big|\,\dots\,\big|\,\ev_{a_t,b_{t,1}}(f),\dots,\ev_{a_t,b_{t,n_t}}(f)\right).
\end{align*} In order to use the map $\ev_{\abf,\bbf}$ for our goal, we need to impose some further conditions on the evaluation parameters $\abf$ and $\bbf$. Aligning with the notation introduced in \cite{NERI2023105703}, we will say that $(\abf,\bbf)\in\Fm^t\times\Fm^\nbf$ is an {\em evaluation pair} if the norms over $\Fq$ of the components of $\abf$ are non-zero and pairwise distinct, and the elements of each component of $\bbf$ are linearly independent over $\Fq$. In particular, note that the last condition is equivalent to require that $w_{\srk}(\bbf)=n$.
\begin{definition}[{\cite{MARTINEZPENAS2018587}}]
    Let $(\abf,\bbf)\in\Fm^t\times\Fm^\nbf$ be an evaluation pair. A {\em linearized Reed--Solomon code over $\Fm^{\mathbf{n}}$ with parameters $\mathbf{a}$, $\mathbf{b}$ and dimension $k$} is defined as \[\LRS_k(\mathbf{a},\mathbf{b})\coloneqq\big\{\ev_{\mathbf{a},\mathbf{b}}(f)\st f\in\ore_{<k}\big\}.\]
\end{definition}
Note that the conditions for $(\abf,\bbf)\in\Fm^t\times\Fm^\nbf$ to form an evaluation pair imply that the maximum number of blocks of a linearized Reed--Solomon is $t=q-1$ and that the length of each block must satisfy $n_i\leq m$, for all $i\in\{1,\dots,t\}$.

It is immediate to see that a generator matrix for $\LRS_k(\mathbf{a},\mathbf{b})$ is defined as the block matrix $G=(G_1|\cdots|G_t)$, with
\begin{align} \label{eq: gen matrix LRS}
    G_i &=\left(\begin{array}{ccc}
    b_{i,1} & \dots & b_{i,n_i}\\
    \sigma(b_{i,1})N_1(a_i) & \cdots & \sigma(b_{i,n_i})N_1(a_i) \\
    \sigma^2(b_{i,1})N_2(a_i) & \cdots & \sigma^2(b_{i,n_i})N_2(a_i) \\
    \vdots & & \vdots \\
    \sigma^{k-1}(b_{i,1})N_{k-1}(a_i) & \cdots & \sigma^{k-1}(b_{i,n_i})N_{k-1}(a_i) \\
    \end{array}\right)  \nonumber \\
    &= \underbrace{\left(\begin{array}{cccc}
    1 & & & \\
    & N_1(a_i) & & \\
    &  & \ddots & \\
    & & & N_{k-1}(a_i) \end{array}\right)}_{\coloneqq D_i} \ \underbrace{\left(\begin{array}{ccc}
    b_{i,1} & \cdots & b_{i,n_i}\\
    \sigma(b_{i,1}) & \cdots & \sigma(b_{i,n_i}) \\
    \vdots & & \vdots \\
    \sigma^{k-1}(b_{i,1}) & \cdots & \sigma^{k-1}(b_{i,n_i}) \\
    \end{array}\right)}_{\coloneqq M_i},
\end{align} for $i\in\{1\dots,t\}$.
\begin{remark}\label{remark: blocks are gen gab}
    Recall that if $t=1$, the sum-rank metric reduces to the notion of rank metric over $\Fmn$. If in addition $a_1=1$, \emph{i.e.}, $\abf=(1)$, the definition of linearized Reed--Solomon coincides with that of generalized Gabidulin codes. More precisely, recall that (generalized) Gabidulin codes are the $q$-analogues of Reed-Solomon codes and are defined as follows (see \cite{kshevetskiy2005new} for further details): let $b_1,\dots,b_n\in\Fm$ be linearly independent elements over $\Fq$ and $s\in\{1,\dots,m-1\}$ be an integer coprime with $m$. The {\em generalized Gabidulin code, $\Gab_{k,s}(b_1,\dots,b_n)$, with evaluation vector $(b_1,\dots,b_n)$, parameter $s$ and of dimension $k$} over $\Fmn$ is generated by the matrix
    \[\begin{pmatrix}
        b_1 &  \cdots & b_n \\ b_1^{[s]} & \cdots & b_n^{[s]} \\ \vdots & & \vdots \\ b_1^{[s(k-1)]} & \cdots & b_n^{[s(k-1)]}
    \end{pmatrix}, \] where $[i]:=q^i$, for any non-negative integer $i$. In particular, we simply say that $\Gab_k\coloneqq\Gab_{k,1}$ is a {\em Gabidulin code of dimension $k$}. \\ Note that the matrix $M_i$ depicted above coincides exactly with a generator matrix for $\Gab_{k,s}$ by picking $\sigma$ to be the automorphism $\sigma:\Fm\longrightarrow\Fm, \ \alpha\longmapsto\alpha^{[s]}$. Moreover, in the representation \eqref{eq: gen matrix LRS} above, the action of the diagonal matrix $D_i$ is nothing but a change of basis for the code generated by $M_i$. Therefore, $G_i$ generates as well a generalized Gabidulin code over $\Fmn$; specifically, $\rowspan_{\Fm}(G_i)=\Gab_{k,s}(b_{i,1},\dots,b_{i,n_i})$.
 
    At the same time, if $n_1=\dots=n_t=1$, $\abf=(1)$ and $\sigma=\id_{\Fm}$, it is immediate to conclude that a $k$-dimensional LRS code reduces to a $k$-dimensional Reed-Solomon code over $\Fmn$.
\end{remark} \medskip

In~\cite{NERI2023105703} the notion of linearized Reed--Solomon codes was further generalized by extending the evaluation of skew polynomials at the points $0$ and $\infty$.
\begin{definition}
    Let $\nbf=(n_1=m,\dots,n_{t-2}=m,n_{t-1}=1,n_t=1)$, $(\abf,\bbf)\in\Fm^{t-2}\times\Fm^{(n_1,\dots,n_{t-2})}$ be an evaluation pair and $b_0,b_\infty\in\Fm^*$. A {\em doubly-extended linearized Reed--Solomon code over $\Fm^{\mathbf{n}}$ with parameters $\mathbf{a}$, $\mathbf{b}$, $b_0$, $b_\infty$ and dimension $k$} is defined as \[\LRS_k(\mathbf{a},\mathbf{b}, b_0, b_\infty)\coloneqq\big\{\big(\ev_{\mathbf{a},\mathbf{b}}(f) \,|\, \ev_{0,b_0}(f) \,|\, \ev_{\infty,b_\infty}(f) \big)\st f\in\ore_{<k}\big\},\]
    where for $f=\sum_if_ix^i\in\ore_{<k}$, $\ev_{\infty,b_\infty}(f)\coloneqq b_\infty f_{k-1}$.
\end{definition}
Note once again that because of the constraints imposed by the definition of an evaluation pair, other than the last two one-length blocks, a doubly-extended linearized Reed--Solomon can have at most $q-1$ blocks of length $m$.
\begin{theorem}(\cite[Theorem 4]{MARTINEZPENAS2018587},\cite[Theorem 4.6]{NERI2023105703})
    Let $\LRS_k\subseteq\Fm^\nbf$ be a (doubly extended) linearized Reed--Solomon code. Then, $\LRS_k$ is MSRD.
\end{theorem}

\section{The Geometry of Linearized Reed--Solomon Codes} \label{The Geometry of Linearized Reed-Solomon Codes}
It is well known that a $k$-dimensional Reed--Solomon code over $\Fq^n$ corresponds to a set of points in $\mathbb{P}^{k-1}(\Fq)$ lying on a {\em rational normal curve} (see ~\cite[p. 120]{arbarello2013geometry}). In the rank-metric setting, MRD codes over $\Fmn$ can be similarly associated with finite geometric objects, namely linear sets scattered with respect to hyperplanes in $\mathbb{P}^{k-1}(\Fm)$ (see~\cite{sheekey2019linear}). However, up to now, no specific geometric characterization of Gabidulin codes was known, except for the work of~\cite{astore2025geometric}. In this section, we investigate the geometry of linearized Reed--Solomon codes, developing a generalization of the connection existing in the Hamming-metric setting and providing a foundation for a deeper geometric understanding of Gabidulin codes, which will be developed further in the next section, specifically in Subsection~\ref{LRS codes with an arbitrary number of blocks}.

Let $j$ be a natural number coprime with $m$. Consider the vectors $\albf\coloneqq(\alpha_1,\alpha_2,\dots, \alpha_k)$, $\bebf\coloneqq(\beta_1,\beta_2,\dots,\beta_k)\in\Z^k$, where 
$$\alpha_i\coloneqq\frac{q^{(i-1)j}-1}{q^j-1}\text{ and }\beta_i\coloneqq\frac{q^{(k-1)j}-1}{q^j-1}-\alpha_i$$
for all $i\in\{1,\dots,k\}$. 
\begin{definition}
   The $q^j$-{\em rational normal curve} $X_{q^j}\subseteq\mathbb{P}^{k-1}(\overline{\F}_{q^m})$ is the curve defined by the monomial map
\begin{align}\label{eq:parametrization_rat_norm_curve}
    \varphi:\,\mathbb{P}^1(\overline{\F}_{q^m})&\longrightarrow\mathbb{P}^{k-1}(\overline{\F}_{q^m}) \\
    [a:b]&\longmapsto[a^{\alpha_1}b^{\beta_1}:\cdots:a^{\alpha_k}b^{\beta_k}], \nonumber
\end{align}
where we use the convention ``$0^0 = 1$''.
\end{definition}
Let $\sigma:\overline{\F}_{q^m}\longrightarrow\overline{\F}_{q^m}$ be the automorphism given by $\sigma(a)=a^{q^j}$ for all $a\in\overline{\F}_{q^m}$. Note that the restriction of $\sigma$ to $\F_{q^m}$ is a generator of the Galois group $\Gal(\Fm/\Fq)$. Using the notion of truncated norm introduced above, the curve $X_{q^j}$ can be equivalently described as the set 
\[X_{q^j}=\left\{\left[N_0(a):N_1(a):\cdots:N_{k-1}(a)\right]:a\in \overline{\F}_{q^m}\right\} \cup \{P_\infty\},\]
where $P_\infty\coloneqq[0:\cdots:0:1]$ and the $N_i$'s are the truncated norms defined in (\ref{eq:partial_norms}) and (\ref{eq:norm}).

The $q^j$-rational normal curve can be realized as a monomial projection of the rational normal curve of degree $\alpha_k$. However, this property is shared by a large family of monomial curves and therefore does not uniquely characterize it. The terminology is instead justified by the theorem below, which provides a determinantal representation of the curve. This representation is given in terms of a matrix of forms that acts as a natural $q^j$-analogue of the classical Hankel matrix defining standard rational normal curves.
\begin{proposition}\label{prop: determinantal variety}
    The vanishing ideal of the $q^j$-rational normal curve $X_{q^j}$ is generated by the $2\times2$ minors of the matrix
    \[\begin{pmatrix}
        x_2 & x_3 & x_4 & \cdots & x_k\\
        x_1^{q^j} & x_2^{q^j} & x_3^{q^j} & \cdots & x_{k-1}^{q^j}
      \end{pmatrix}.\] Moreover, these minors are linearly independent
    homogeneous polynomials of degree $q^j+1$.
\end{proposition}
\begin{proof}
  Let us first prove the linear independence of the generators. The minors have the shape:
  \begin{equation}\label{eq:binomials}
    P_{uv} \coloneqq x_{u}x_{v}^{q^j} - x_{u-1}^{q^j}x_{v+1}, \qquad \text{for}\qquad 2 \leq u \leq v \leq k-1.
  \end{equation}
  Observe that no monomial can occur in two distinct minors above, which
  proves the independence of the $P_{uv}$'s when $2 \leq u \leq v \leq k-1$.

  Let $I$ be the ideal generated by the $2\times 2$ minors of the matrix displayed above and let $Z(I)$ be the variety defined by $I$. Consider a point $[a_1:\dots:a_k]\in Z(I)(\overline{\F}_{q^m})$. If $a_1=0$, we immediately conclude that $a_2=\dots=a_{k-1}=0$. Hence, the only point obtained in this case is $P_\infty$. If $a_1=1$, considering the defining equations involving $x_1$, we obtain $a_i=N_{i-1}(a_2)$, for all $i\in\{1,\dots,k\}$. As a consequence, we conclude that $Z(I)$ coincides with $X_{q^j}$ when regarded as varieties over $\overline{\F}_{q^m}$.
  
  Therefore, to show that $I=I(X_{q^j})$, it only remains to prove
  that $I$ is radical. Equivalently, we need to prove that the
  subscheme $\mS$ of $\mathbb{P}^{k-1}$ defined by $I$ is reduced.
  Suppose it is not, since the the underlying
  reduced scheme $\mS_{\text{red}}$ (see \cite[Proposition~3.50]{GW20} for a definition)
  is $X_q$ which is irreducible, then $\mS$ would have no regular point. Thus,
  we are reduced to prove the existence of a regular point for $\mS$. Namely
  we will prove that $\mS$ is regular at $P_0 \coloneqq [1 : 0 : \cdots : 0]$.
  From \cite[Corollary~6.29]{GW20}, $\mS$ is regular at $P_0$ if and only if
  its tangent space at $P_0$ has dimension $1$ and from \cite[Proposition~6.10(2)]{GW20}
  this tangent space can be obtained from the kernel of the Jacobian matrix of the $P_{uv}$'s introduced in (\ref{eq:binomials}). Since
  \[
    {\left(\frac{\partial P_{uv}}{\partial x_i}(P_0)\right)}_{i=1}^k = \left\{
      \begin{array}{cc}
        (0, \dots, 0, \underbrace{-1}_{v+1}, 0 \dots, 0) & \text{if}\ u=2,\\
        (0,\dots, 0) & \text{otherwise},
      \end{array}
      \right.
    \]
    the tangent space of $\mS$ at $P_0$ is the projectivisation of the kernel of
    the matrix
    \[
      \begin{pmatrix}
        0 & 0 & 1 & 0 & \cdots & 0 \\
        0 & 0 & 0 & \ddots &  \ddots & \vdots \\
        0 & 0 & \vdots & \ddots & \ddots   & 0 \\
        0 & 0 & 0 & \dots  & 0 & 1
      \end{pmatrix}.
    \]
    The above matrix has a kernel of dimension $2$, and hence the projectivisation of the
    kernel has dimension $1$, which, from \cite[Proposition~6.10(2)]{GW20}, shows that the
    tangent space of $\mS$ at $P_0$ has dimension $1$ and hence that $\mS$ is regular at $P_0$.
\end{proof}
\begin{example}
    To illustrate the previous proposition, let us set $j=1$ and $k=4$. Then, $\sigma$ is simply the Frobenius automorphism and $N_i(a)=a^{\frac{q^i-1}{q-1}}$, for all $i\geq0$ and $a\in\overline{\F}_{q^m}$. Consider the ideal $I$ generated by the $2\times2$ minors of the matrix
    \[\begin{pmatrix}
        x_2 & x_3 & x_4\\x_1^q & x_2^q & x_3^q
    \end{pmatrix},\] 
    and its associated variety 
    \begin{equation*}
        Z(I)=\{[x_1:x_2:x_3:x_4]\in\mathbb{P}(\overline{\F}_{q^m})\st x_2^{q+1}-x_1^qx_3=0, \ x_3^{q+1}-x_2^qx_4=0, \ x_2x_3^q-x_4x_1^q=0\}.
    \end{equation*}
    Let $[a_1:a_2:a_3:a_4]$ be a point in $Z(I)$. If $a_1=0$, we immediately get that $a_2=a_3=0$, obtaining the point $[0:0:0:1]=P_\infty$. Otherwise,
        \[\begin{cases}
            x_1=1 \\
            x_3=x_2^{q+1} \\
            x_3^{q+1}=x_2^qx_4 \\
            x_4=x_2x_3^q \\
        \end{cases}\implies \
        \begin{cases}
            x_1=1 \\
            x_3=x_2^{q+1}=x_2^{\frac{q^2-1}{q-1}} \\
            x_4=x_2^{q^2+q+1}=x_2^{\frac{q^3-1}{q-1}} \\
        \end{cases}.\] 
        Hence, we get all the points of the form \[\left[1:x_2:x_2^{\frac{q^2-1}{q-1}}:x_2^{\frac{q^3-1}{q-1}}\right]=\left[N_0(a_2):N_1(a_2):N_2(a_2):N_3(a_2)\right],\] for $a_2\in\overline{\F}_{q^m}$. In conclusion, $$Z(I)=\left\{\left[N_0(a):N_1(a):N_2(a):N_3(a)\right]:a\in \overline{\F}_{q^m}\right\} \cup \{P_\infty\}.$$
\end{example} \smallskip

Recall that an \emph{$\Fm$-rational point} of $\mathbb{P}^{k-1}$ is a point whose homogeneous coordinates are proportional to a vector in $\F_{q^m}^k$. Then, the {\em $\Fm$-rational points of $X_{q^j}$} are \[X_{q^j}(\Fm)=\left\{\left[N_0(a):N_1(a):\cdots:N_{k-1}(a)\right]:a\in \Fm\right\} \cup \{P_\infty\}.\] In particular, observe that $\lvert X_{q^j}(\Fm)\rvert=q^m+1$.

For all $a\in\Fm$, we denote by $\mV_a(\Fm)$ the set
\begin{equation*}
    \mV_a(\Fm)\coloneqq \left\{\left[b:N_1(a)\sigma(b):N_2(a)\sigma^2(b):\cdots:N_{k-1}(a)\sigma^{k-1}(b)\right]:b\in\Fm^*\right\}.
\end{equation*}

\begin{lemma} \label{lemma: different norms different sets}
  Let $a_1,a_2\in\Fm$. Then, the following statements are equivalent
  \begin{enumerate}
  \item\label{item:norms} $N_m(a_1)=N_m(a_2)$.
  \item\label{item:same_sets} $\mV_{a_1}(\Fm)=\mV_{a_2}(\Fm)$;
  \item\label{item:intersection} $\mV_{a_1}(\Fm) \cap \mV_{a_2}(\Fm) \neq \emptyset$;
  \end{enumerate}
\end{lemma}
\begin{proof}
(\ref{item:norms}) $\Rightarrow$ (\ref{item:same_sets}).    Let $a_1,a_2\in\Fm$ be such that $N_m(a_1)=N_m(a_2)$. By Hilbert 90 theorem over finite fields (see~\cite[Exercise 2.33]{lidl1994introduction}), this equality holds if and only if $a_1=a_2\frac{\sigma(u)}{u}$ for some $u\in\Fm^*$. As a consequence, for all $b\in\Fm^*$, it holds that
    \begin{gather*}
        \left[b:N_1(a_1)\sigma(b):\cdots:N_{k-1}(a_1)\sigma^{k-1}(b)\right] \\[-6pt] \verteq \\[-4pt]
        \left[b:N_1(a_2)N_1\left(\frac{\sigma(u)}{u}\right)\sigma(b):\cdots:N_{k-1}(a_2)N_{k-1}\left(\frac{\sigma(u)}{u}\right)\sigma^{k-1}(b)\right] \\[-4pt] \verteq \\[-4pt]
        \left[b:N_1(a_2)\frac{\sigma(u)}{u}\sigma(b):\cdots:N_{k-1}(a_2)\frac{\sigma(u)\sigma^2(u)\cdots\sigma^{k-1}(u)}{u\,\sigma(u)\cdots\sigma^{k-2}(u)}\sigma^{k-1}(b)\right] \\[-4pt] \verteq \\[-4pt]
        \left[b:\frac{N_1(a_2)\sigma(ub)}{u}:\cdots:\frac{N_{k-1}(a_2)\sigma^{k-1}(ub)}{u}\right] \\[-4pt] \verteq \\[-4pt]
        \left[ub:N_1(a_2)\sigma(ub):\cdots:N_{k-1}(a_2)\sigma^{k-1}(ub)\right]
      \end{gather*} which yields the desired implication.

      (\ref{item:same_sets}) $\Rightarrow$ (\ref{item:intersection}) is obvious.
      
      (\ref{item:intersection}) $\Rightarrow$ (\ref{item:norms}). Let
      $a_1,a_2\in\Fm$ and $b_1, b_2\in\Fm^*$ such that
      \[\left[b_1:N_1(a_1)\sigma(b_1):\cdots:N_{k-1}(a_1)\sigma^{k-1}(b_1)\right]=\left[b_2:N_1(a_2)\sigma(b_2):\cdots:N_{k-1}(a_2)\sigma^{k-1}(b_2)\right].\]
      In particular,
      $N_1(a_1)\sigma(b_1)b_1^{-1}=N_1(a_2)\sigma(b_2)b_2^{-1}$, that
      is
      $a_1\sigma(b_1)b_1^{-1}=a_2\sigma(b_2)b_2^{-1}$. Hence,
      \[a_1=a_2\frac{\sigma(b_2)\sigma(b_1)^{-1}}{b_2b_1^{-1}}=a_2\frac{\sigma\left(b_2b_1^{-1}\right)}{b_2b_1^{-1
          }},\] and we conclude that $N_m(a_1)=N_m(a_2)$ using Hilbert 90 theorem for finite fields the other way around.
\end{proof}

For the sake of completeness, we restate and prove the following result, originally appearing as \cite[Proposition 3.3]{donati2018generalization} and \cite[Remark 4.5]{santonastaso2023subspace}, in the notation of the present paper.
\begin{lemma} \label{lemma: LRS partition normal curve}
    Let $\mathcal{N}$ be a set of $q$ elements in $\Fm$ with pairwise distinct norms. Then, \[X_{q^j}(\Fm)=\bigsqcup_{a\in \mathcal{N}}\mV_a(\Fm)\sqcup \{P_\infty\}.\]
\end{lemma}
\begin{proof}
  Let $\mV_\mN$ denote the set $\bigcup_{a\in \mathcal{N}}\mV_a(\Fm)$. Note that by definition $\lvert \mV_a(\Fm)\rvert=\frac{q^m-1}{q-1}$, if $a\in\mathcal{N}\setminus\{0\}$, while $\lvert \mV_0(\Fm)\rvert=1$. From Lemma~\ref{lemma: different norms different sets} the sets $\mV_a(\Fm)$ are pairwise disjoint for $a\in\mN$, and \[\lvert\mV_\mN\rvert=(q-1)\frac{q^m-1}{q-1}+1=q^m.\]
  At the same time, by the definition of $X_{q^j}$ it immediately follows that $X_{q^j}(\Fm)\subseteq \mV_\mN\sqcup\{P_\infty\}$. Therefore, since $|\,X_{q^j}(\Fm)\,|=q^m+1$, we conclude with a counting argument.
\end{proof}

The $\Fm$-rational points of the curve $X_{q^j}$ are therefore partitioned into linear sets associated with Gabidulin codes (together with two distinguished points). In the next proposition, we establish a stronger result: if the linear set arising from an MRD code is contained in the curve, then it must necessarily coincide with one of the components of this partition.
\begin{proposition}\label{lemma: onlygab}
    Let $k\geq3$ and let $\mC\subseteq\Fm^m$ be a $k$-dimensional MRD code whose associated set of points $\mL_{\mU_\mC}$ is contained in $X_q$. Then, there exists $a\in\Fm^*$ such that $\mL_{\mU_\mC}=\mathcal{V}_a(\Fm)$ and $\mC$ is a Gabidulin code.
\end{proposition}
\begin{proof}
  By \cite[Corollary 2.7]{astore2025geometric}, any $k$-dimensional MRD code $\mC$ satisfies $\dim{(I(\mL_{\mU_\mC})}_{q+1})\leq\binom{k-1}{2}$. Moreover, $\mC$ is a Gabidulin code if and only if $\dim{(I(\mL_{\mU_\mC})}_{q+1})=\binom{k-1}{2}$. Let $\mC\subseteq\Fm^m$ be an MRD code whose associated set of points $\mL_{\mU_\mC}$ is contained in $X_q$. Then, ${I(X_q)}_{q+1}\subseteq {I(\mL_{\mU_\mC})}_{q+1}$, which implies $\dim({I(\mL_{\mU_\mC})}_{q+1})\geq\binom{k-1}{2}$ by Proposition~\ref{prop: determinantal variety}. Therefore, we conclude that $\mC$ is a Gabidulin code and ${I(\mL_{\mU_\mC})}_{q+1}={I(X_q)}_{q+1}$.

  In particular, there exists a matrix $A\in\mathrm{PGL}_k(\Fm)$ that maps $\mathcal{V}_1(\Fm)$, which is the set of points associated with a standard Gabidulin code over $\Fm^m$, to $\mL_{\mU_\mC}$. By the argument above, we know that ${I(AX_q)}_{q+1}\subseteq {I(\mL_{\mU_\mC})}_{q+1}={I(X_q)}_{q+1}$, and since the Hilbert sequence is invariant under change of coordinates, we obtain ${I(AX_q)}_{q+1}={I(X_q)}_{q+1}$ since both spaces have the same dimension. Since both ideals are generated in degree $q+1$, we conclude that $I(AX_q)=I(X_q)$, and therefore $AX_q=X_q$. Hence, the matrix $A$ induces an automorphism of $X_q$. Since $\mathbb{P}^1$ is the normalization of $X_q$ (see \cite[Section~2.5.2]{S13} for a definition), from \cite[Theorem~2.21]{S13} and its subsequent corollary, the map $A$ canonically lifts to an automorphism of $\mathbb{P}^1$, \emph{i.e.}, there exists an automorphism $\tilde A:\mathbb{P}^1\longrightarrow\mathbb{P}^1$ such that the following diagram commutes 
    \begin{equation}\label{eq:diagram}
    \begin{tikzcd}
    \mathbb{P}^1 \arrow[r, "\tilde{A}"] \arrow[d, "\varphi"'] & \mathbb{P}^1 \arrow[d, "\varphi"] \\
    X_q \arrow[r, "A"']                                     & X_q
  \end{tikzcd}\ , \qquad \text{where}\ \varphi : \mathbb{P}^1 \rightarrow
  X_q\  \text{is the map \eqref{eq:parametrization_rat_norm_curve}}.
\end{equation}
Being an automorphism of $\mathbb{P}^1$, such map can therefore be represented by a matrix $\tilde{A}={(\tilde a_{ij})}\in \mathrm{PGL}_2(\Fm)$ as $[x:y] \mapsto [x:y]\cdot\tilde A$. Since the only singular point of $X_q$ is $P_{\infty}=[0:\dots:0:1]$ (a fact we prove in Proposition~\ref{prop:hilbpolyofS} to avoid a technical detour here), it follows that $A$ must fix $P_{\infty}$. Moreover, $P_{\infty} = \varphi([1:0])$, and hence $[1:0] = [1:0]\cdot\tilde A$, which yields $\tilde a_{12}=0$.

From the commutativity of the diagram \eqref{eq:diagram}, using the notation of \eqref{eq:parametrization_rat_norm_curve}, we have that for any $[u:v] \in \mathbb{P}^1(\overline{\F}_{q^m})$,
\begin{equation}\label{eq:polynomial_identity}
  [u^{\alpha_1}v^{\beta_1} : \cdots : u^{\alpha_k}v^{\beta_k}] \cdot A = 
  [(\tilde{a}_{11}u+\tilde{a}_{21}v)^{\alpha_1}(\tilde{a}_{22}v)^{\beta_1} : \cdots : (\tilde{a}_{11}u+\tilde{a}_{21}v)^{\alpha_k}(\tilde{a}_{22}v)^{\beta_k}].
\end{equation}
Therefore, the linear span of the monomials of the left-hand side must coincide with that of the polynomials of the right-hand side. In particular, the former contains  $(\tilde a_{11}u+\tilde a_{12}v)^{\alpha_k}$. If $\tilde a_{12}\neq 0$, this would determine monomials that do not appear in the linear span of the monomials of the left-hand side of (\ref{eq:polynomial_identity}).
Hence, we conclude that $\tilde A \in \mathrm{PGL}_2(\Fm)$ can be represented as
     \begin{equation*}
        \tilde A=\begin{pmatrix}
            a & 0 \\ 0 & 1
        \end{pmatrix},
    \end{equation*}
    for some $a\in\Fm^*$. This shows that $A$ sends $\mathcal{V}_1(\Fm)$ onto $\mathcal{V}_a(\Fm)$, which henceforth equals $\mL_{\mU_\mC}$.
\end{proof}
\begin{remark}
    We point out that the assumption $k\geq3$ in the previous lemma is a necessary condition. Indeed, for $k=2$ the curve $X_q$ is just $\mathbb{P}^1(\overline{\F}_q)$. This implies that for every $2$-dimensional MRD code, independently on the choice of the generator matrix, the associated set of points $\mathcal{L}_{\mathcal{U}_{\mC}}$ is contained in $X_q$.
\end{remark}
We conclude this section by showing that the associated set of projective points of a doubly-extended linearized Reed--Solomon code lies on a $q^j$-rational normal curve.
\begin{theorem} \label{theorem: linear set of lrs}
    Let $\mC$ be a (doubly-extended) linearized Reed--Solomon code over $\Fm^{\mathbf{n}}$, defined as an evaluation code over the skew polynomial ring $\ore$ with $\sigma: a\longmapsto a^{q^j}$. Then, there exists a generator matrix $G=(G_1|\cdots|G_t)\in\Fm^{k\times \mathbf{n}}$ of $\mC$ such that for every $i\in\{1,\dots,t\}$, either $\mL_{\mU_i}=\{P_{\infty}\}$, where $\mU_i\coloneqq\colspan_{\Fq}(G_i)$, or there exists $a_i\in\Fm$ such that $\mL_{\mU_i}\subseteq\mV_{a_i}(\Fm)$ and $N_m(a_i)\neq N_m(a_j)$ for $i\neq j$. As a consequence, $$\mL_G=\bigsqcup_{i=1}^t\mL_{\mU_i}\subseteq X_{q^j}(\Fm).$$ In particular, if $\mC$ is a doubly-extended linearized Reed--Solomon code with $q+1$ blocks and parameters $n_1=\dots=n_{q-1}=m$, then there exists a generator matrix $G$ of $\mC$ for which $\mL_G=X_{q^j}(\Fm)$ and there exist elements $a_1,\dots,a_{q-1}\in\Fm^*$ with pairwise distinct norms such that $\mL_{\mU_i}=\mV_{a_i}(\Fm)$, for $i\in\{1,\dots,q-1\}$.
\end{theorem}
\begin{proof}
    Let $\mC$ be a linearized Reed--Solomon code over $\Fm^{\mathbf{n}}$ and consider the generator matrix described in~\eqref{eq: gen matrix LRS} for some evaluation pair $(\abf,\bbf)\in\Fm^t\times\Fm^\nbf$. Recalling that the entries of $\abf$ have pairwise distinct norms over $\Fq$, it immediately follows that $$\mL_{\mU_i}\subseteq\mV_{a_i}(\Fm),$$ for all $i\in\{1,\dots,t\}$, where the equality holds if $n_i=m$. On the other hand, if $\mC$ is doubly-extended, then the previous result still holds for all $i\in\{1,\dots,t-2\}$, while $\mL_{\mU_{t-1}}=\{[1:0:\cdots:0]\}$ and $\mL_{\mU_t}=\{P_{\infty}\}$. In both cases, Lemma~\ref{lemma: different norms different sets} implies $\mL_{\mU_i}\cap\mL_{\mU_j}=\emptyset$ if $i,j\in\{1,\dots,t\}$ and $i\neq j$. \\ Finally, if $\mC$ is a doubly-extended linearized Reed--Solomon code with $q+1$ blocks and parameters $n_1=\dots=n_{q-1}=m$, then $\mL_{\mU_i}=\mV_{a_i}$ for all $i\in\{1,\dots,q-1\}$. Therefore, by Lemma~\ref{lemma: LRS partition normal curve} we get $\lvert \mL_G\rvert=q^m+1$ and hence the desired equality.
\end{proof}

\begin{remark}
    In~\cite[Remark 4.8]{NERI2023105703}, it was observed that the linear set associated with a doubly-extended linearized Reed--Solomon code of dimension 2 is the projective line. Theorem~\ref{theorem: linear set of lrs} extends this result to arbitrary dimensions. In particular, for $k=2$, the $q^j$-rational normal curve coincides with the projective line.
\end{remark}

\subsection{Characterization of Linearized Reed--Solomon Codes}

In this subsection, we provide a characterization of linearized Reed--Solomon codes of dimension $k\leq m$, building on the geometric interpretation of this class of codes. This result yields an effective method to determine whether a given generator matrix of a sum-rank metric code defines a linearized Reed--Solomon code, and hence a practical way to differentiate them from non-structured codes.

\begin{theorem} \label{theorem: characterization LRS}
    Let $k\in\{3,\dots,m\}$, $t\leq q-1$ and assume that $n_i=m$ for all $i\in\{1,\dots,t\}$. Let $G=(G_1|\cdots|G_t)$ be a generator matrix of the $[\nbf,k]_{q^m/q}$ code $\mC$. Let $\mC_i=\rowspan_{\Fm}(G_i)$ and $\mU_i=\colspan_{\Fq}(G_i)$. Then, $\mC$ is a linearized Reed--Solomon code, defined as an evaluation code over the skew polynomial ring $\ore$ with $\sigma: a\longmapsto a^q$, if and only if all the following hold for all $i,j\in\{1,\dots,t\}$ with $i\neq j$:
    \begin{enumerate}
        \item $\mC_i$ is a $k$-dimensional Gabidulin code over $\Fm^m$;
        \item $\mL_{\mU_i}\,\cap\,\mL_{\mU_j}=\emptyset$;
        \item ${I(\mL_{\mU_i})}_{q+1}={I(\mL_{\mU_j})}_{q+1}$.
    \end{enumerate}
\end{theorem}
\begin{proof}
 First, assume that $\mC$ is an $[\nbf,k]_{q^m/q}$ linearized Reed--Solomon code. As shown in Remark~\ref{remark: blocks are gen gab}, the first condition is a direct consequence of the definition of linearized Reed--Solomon codes, while the second relation has already been proved in Theorem~\ref{theorem: linear set of lrs}. For the last point, note that changing the basis of the subcode $\mC_i$ does not affect the geometric properties of its associated object $\mL_{\mU_i}$ for all $i\in\{1,\dots,t\}$. Hence, we know from \cite[Corollary 3.6]{astore2025geometric} that \[\dim I(\mL_{\mU_i})_{q+1}=\binom{k-1}{2}.\] Moreover, since $\mL_{\mU_i}\subseteq\mL_G$, it also holds that $I(\mL_G)_{q+1}\subseteq I(\mL_{\mU_i})_{q+1}$. Since $I(\mL_G)$ contains all the $\binom{k-1}{2}$ linearly-independent polynomials $\{x_{u}x_{v}^q - x_{u-1}^q x_{v+1} \st 2 \leq u \leq v \leq k-1\}$ determined in Proposition~\ref{prop: determinantal variety}, we conclude by a counting argument that $I(\mL_G)_{q+1}=I(\mL_{\mU_i})_{q+1}$. \\ Let us now assume that all the above conditions hold. Up to a change of coordinate, we set $\mathcal{L}_{\mathcal{U}_1}=\mathcal{V}_1(\Fm)$. The third condition states that $I(\mL_{\mU_G})_{q+1}=I(\mL_{\mU_i})_{q+1}$. In particular, this implies that $\mC$ can be represented as a concatenation of $t$ Gabidulin codes, whose associated sets of projective points $\mL_{\mU_i}$ lie on the same rational normal curve $X_q(\Fm)$. Then, by Proposition~\ref{lemma: onlygab}, the first condition implies that $\mL_{\mU_i}=\mV_{a_i}(\Fm)$ for some $a_i\in\Fm^*$ and for all $i\in\{1,\dots,t\}$. To conclude it is enough to note that the second condition ensures that the norms of the elements $a_i$, with $i\in\{1,\dots,t\}$, are all pairwise distinct by means of Lemma~\ref{lemma: different norms different sets}. Due to Lemma~\ref{lemma: LRS partition normal curve} and Theorem~\ref{theorem: linear set of lrs}, it then follows that \[\mL_{\mU_G}=\bigsqcup_{i=1}^t\mV_{a_i}(\Fm)=X_q(\Fm)\setminus\{P_0,P_\infty\},\] hence the corresponding code $\mC$ is a linearized Reed--Solomon code.
\end{proof}

\begin{remark} \label{remark:concatenation not LRS}
    Note that the theorem's second condition directly holds if $\mC$ is known to be MSRD, see \cite[Section 5]{NERI2023105703}. Moreover, the third condition ensures that the linear sets associated with each component of $\mC$ all lie on the same $q$-rational normal curve, instead of being arbitrarily distributed in $\mathbb{P}^{k-1}(\Fm)$. In light of Theorem~\ref{theorem: linear set of lrs}, this corresponds to excluding the possibility that $\mC$ is a simple concatenation of Gabidulin codes, with no relations of the form~\eqref{eq: gen matrix LRS} above connecting them.
    
    At the same time, the third condition can be further interpreted from a purely coding-theoretic point of view. By the theorem's proof and \eqref{eq: Fq-Hilbert seq}, it is immediate to conclude that, if $\mC$ is a linearized Reed--Solomon code whose parameters satisfy the hypothesis above, then \[\dim{\left({\mC^{\mathrm{H}}}\right)}^{(q+1)} = \binom{k+q}{q+1}-\binom{k-1}{2},\] where we recall that $\mC^{\mathrm{H}}$ denotes the associated Hamming-metric code with $\mC$.
\end{remark}

\begin{remark}
    From a computational point of view, the conditions stated in Theorem~\ref{theorem: characterization LRS} can be checked in \emph{polynomial time} with respect to the defining parameters of the code. More precisely, the first condition can be verified in polynomial time using the most suitable characterization of Gabidulin codes among those listed in \cite[Theorem 6.5]{NERI2020418}. The second condition holds if $\dim_{\Fq}(\mU_i\, \cap \, \mU_j)=0$ for all $i,j\in\{1,\dots,t\}$ with $i\neq j$. Since this is equivalent to determining whether $\dim_{\Fq}(\mU_i\, \oplus \, \mU_j)=n_i+n_j$ for the same index choices, this step requires the computation of $\binom{t}{2}$ ranks. Assuming that the first two conditions have been verified, we can adopt the following strategy for the third one. Given a generator matrix $G=(G_1|\cdots|G_t)$ of the code under analysis, there exists some matrix $A_1\in\GL_k(\Fm)$ such that $A_1G_1$ is in systematic form. Then, by \cite[Corollary 3.6]{astore2025geometric}, we know that \[I\Big(\mL_{\mU_{A_1G_i}}\Big)_{q+1}=I\Big(\mL_{\mU_{A_1G_1}}\Big)_{q+1}=\underbrace{\langle x_{u}x_{v}^q - x_{u-1}^q x_{v+1}\st 2\leq u \leq v \leq k-1 \rangle_{\Fm}}_{\coloneqq\mP_1},\] for all $i\in\{2,\dots,t\}$. At the same time, for all $i\in\{2,\dots,t\}$, there exists some matrix $A_i\in\GL_k(\Fm)$ such that $A_iA_1G_i$ is in systematic form. Hence, the $(q+1)$-th part of the vanishing ideal of $\mL_{\mU_{A_iA_1G_1}}$ is exactly $\mP_i=\mP_1$. Let $A_i=(a_{j,l})_{j,l\in\{1,\dots,k\}}$ and consider the change of variables $x_j\longmapsto\bar x_j$ such that $\bar x_j=a_{1,j}x_1+\dots+a_{k,j}x_k$ for $j\in\{1,\dots,k\}$. Then, we can derive from $\mP_1$ a new vector space $\mP_{1,i}$ of polynomials in $\bar x_1,\dots,\bar x_k$. Therefore, the third condition of Theorem~\ref{theorem: characterization LRS} holds if $\mP_i=\mP_{1,i}$ for all $i\in\{2,\dots,t\}$. Determining the matrices $A_1,\dots,A_t$ required to put the subgenerator matrices into systematic form is essentially equivalent to performing $t$ Gaussian eliminations of $k\times m$ matrices over $\Fm$. Similarly, verifying whether the equalities $\mP_i=\mP_{1,i}$ hold for all $i\in\{2,\dots,t\}$ is straightforward, since it amounts to checking whether two vector spaces, given by their bases, are equal.
\end{remark}

\section{The Hilbert Sequence of a Linearized Reed--Solomon Code} \label{The Hilbert Sequence of a Linearized Reed-Solomon Code}

The aim of this section is to advance further in the identification of invariants for linearized Reed--Solomon codes by studying the homogeneous forms that vanish on their associated geometric object. As observed in Remark~\ref{remark:concatenation not LRS}, the homogeneous forms of degree $q+1$ play a key role in the characterization of linearized Reed--Solomon codes. It is then natural to investigate whether forms of higher degrees could reveal other distinctive information about the code under analysis. In particular, in this section we study the $(q^m/q)$-Hilbert sequence of doubly-extended linearized Reed--Solomon codes with parameters $t=q+1$ and $\nbf=(m,\dots,m,1,1)$. By Proposition~\ref{prop: H seq and CM reg}, this is equivalent to studying the Hilbert function of the coordinate ring of their associated set of projective points. In the previous section, we showed that this set coincides with $X_{q^j}(\Fm)$, the set of $\Fm$-rational points of the $q^j$-rational normal curve. Now, our goal is to study its Hilbert function.

To simplify the computations, from now on we will only consider the Frobenius automorphism $\sigma$ ($j=1$), but with due care our techniques also apply to any generator of the Galois group $\Gal(\Fm/\Fq)$. Moreover, if $k=2$, we simply have that $X_q=\mathbb{P}^1(\overline{\F}_{q^m})$, while, for $k=3$, $X_q$ reduces to a plane curve. In particular, since it is generated by a single polynomial of degree $q+1$, the Hilbert function is given by the formula $$\HF(d)=\binom{2+d}{2}-\binom{2+d-(q+1)}{2},$$ where the second term is meant to be $0$ if $2+d-(q+1)<2$. Moreover, we immediately conclude that the Hilbert regularity is $q$. Therefore, from now on, we assume $k\in\{4,\dots,m\}$.

In the following, we denote by $$S\coloneqq\Fm[x_1,\dots,x_k]/I(X_q) \ \text{ and } \ S_{\Fm}\coloneqq\Fm[x_1,\dots,x_k]/I(X_q(\Fm))$$ the homogeneous coordinate rings of $X_q$ and $X_q(\Fm)$, respectively. Since $X_q(\Fm)$ is contained in $X_q$, we immediately obtain $\HF_S(d)\geq\HF_{S_{\Fm}}(d)$ for all $d\geq0$. Even though the two Hilbert functions will eventually diverge from each other, we are interested in understanding until which point their equality holds. 

In Proposition~\ref{prop: determinantal variety} we proved that $X_q$ is a toric variety (see \cite[Chapter 4]{sturmfels1996grobner}). Therefore, it follows by~\cite[Lemma 4.1]{sturmfels1996grobner} that $I(X_q)$ is spanned as an $\Fm$-vector space by the set $\{x^\mathbf{u}-x^\mathbf{v}:\mathbf{u},\mathbf{v}\in\Z_{\geq0}^k\text{ such that }A\mathbf{u}=A\mathbf{v}\}$, where
\begin{equation} \label{eq: matrix A}
    A\coloneqq\begin{pmatrix}
    \albf \\
    \bebf
    \end{pmatrix} =\begin{pmatrix}
        0 & 1 & \frac{q^2-1}{q-1} & \cdots &\frac{q^{k-1}-1}{q-1} \\
        \frac{q^{k-1}-1}{q-1} & \frac{q^{k-1}-1}{q-1}-1 & \frac{q^{k-1}-1}{q-1}-\frac{q^2-1}{q-1} & \cdots & 0
    \end{pmatrix}.\end{equation}
In the next lemma, we provide a similar characterization for $I(X_q(\Fm))$.
\begin{lemma}\label{lemma:fbasis}
    The vanishing ideal of $X_q(\Fm)$ is a binomial ideal and is generated as an infinite dimensional $\overline{\F}_{q^m}$-vector space by all homogeneous binomials $x^{\gabf_1}-x^{\gabf_2}$, with $\gabf_1,\gabf_2\in\Z_{\geq0}^k$ that vanish on $P_0=[1:0:\cdots:0]$ and $P_\infty=[0:\cdots:0:1]$, and $\langle \gabf_1,\albf\rangle=\langle \gabf_2,\albf\rangle\pmod{q^m-1}$, where $\albf=\left(0,1,\frac{q^2-1}{q-1},\dots,\frac{q^{k-1}}{q-1}\right)$.
\end{lemma}
\begin{proof}
    The variety $X_q(\Fm)$ is set-theoretically obtained by intersecting $X_q$ and $\mathbb{P}^{k-1}(\Fm)$. Moreover, by~\cite[Corollary 2.6]{beelen2019vanishing} the ideal $I(\mathbb{P}^{k-1}(\Fm))$ is generated by the Fermat polynomials, hence it is also a binomial ideal. Finally, since the sum of binomial ideals is binomial and the radical of a binomial ideal is binomial~\cite[Theorem 3.1]{eisenbud1996binomial}, we conclude that $I(X_q(\Fm))$ is a binomial ideal.
    
    Since $I(X_q(\Fm))$ is a binomial ideal, it is generated as an infinite dimensional $\Fm$-vector space by the set of its homogeneous binomials. Fix a homogeneous binomial $x^{\gabf_1}+bx^{\gabf_2}\in I(X_q(\Fm))$, with $\gabf_1,\gabf_2\in\Z_{\geq0}^k$. By construction, it vanishes on 
    $$X_q(\Fm)=\left\{\left[N_0(a):N_1(a):\cdots:N_{k-1}(a)\right]:a\in \Fm\right\} \cup \{P_\infty\}.$$
    As a consequence,
    $$a^{\langle \gabf_1,\albf\rangle}+ba^{\langle \gabf_2,\albf\rangle}=0\text{, for all }a\in\Fm,$$
    which implies
    $$\langle \gabf_1,\albf\rangle=\langle \gabf_2,\albf\rangle\pmod{q^m-1}\text{ and }b=-1,$$ for $\albf=\left(0,1,\frac{q^2-1}{q-1},\dots,\frac{q^{k-1}}{q-1}\right)$. With a similar argument, one can check that any homogeneous binomial $x^{\gabf_1}-x^{\gabf_2}$, with $\gabf_1,\gabf_2\in\Z_{\geq0}^k$ that vanish on $P_0$ and $P_\infty$, satisfying $\langle \gabf_1,\albf\rangle=\langle \gabf_2,\albf\rangle\pmod{q^m-1}$ belongs to $I(X_q(\F_{q^m}))$.
\end{proof}
\begin{proposition}\label{proposition:coincidences}
    Let $D$ be the smallest integer for which there exist $\gabf_1,\gabf_2\in\Z^{k}_{\geq0}$ with $\lvert \gabf_1\rvert=\lvert\gabf_2\rvert=D$, $\langle \gabf_1,\albf\rangle=\langle \gabf_2,\albf\rangle\pmod{q^m-1}$, and $\langle \gabf_1,\albf\rangle\neq\langle \gabf_2,\albf\rangle$. Then, for all $d<D$ it holds that \[\HF_S(d)=\HF_{S_{\Fm}}(d).\]
\end{proposition}
\begin{proof}
    The statement immediately follows from Lemma~\ref{lemma:fbasis} and the paragraph just above it.
\end{proof}
\begin{remark}\label{remark:Dexists}
    Note that $D$ is always a natural number. Suppose that there are no $\gabf_1,\gabf_2\in\Z^{k}_{\geq0}$ with $\lvert \gabf_1\rvert=\lvert\gabf_2\rvert$, $\langle \gabf_1,\albf\rangle=\langle \gabf_2,\albf\rangle\pmod{q^m-1}\}$, and $\langle \gabf_1,\albf\rangle\neq\langle \gabf_2,\albf\rangle$. This implies that $X_q=X_q(\Fm)$, but this is a contradiction since $X_q$ one has dimension 1 while the other has dimension $0$.
\end{remark}
Although we do not have a closed formula for $D$, we can easily establish a lower bound.
The Hilbert functions of $S$ and $S_{\Fm}$ coincide up to the smallest degree $D$ for which $I(X_q(\Fm))_D\setminus I(X_q)_D$ becomes non-empty. If this is the case, \emph{i.e.}, there exists a polynomial $p(x_1,\dots,x_k)$ of degree $D$ in $I(X_q(\Fm))\setminus I(X_q)$, then Bézout's theorem \cite[Theorem 18.3]{harris2013algebraic} implies \[|\,X_q(\Fm)\,|\leq |\,V(p)\cap X_q \,| \leq \deg(p)\deg(X_q),\] where $V(p)$ is the variety of points in $\mathbb{P}^{k-1}(\Fm)$ that annihilate $p$. Therefore, \[\deg(p)\geq \frac{|\,X_q(\Fm)\,|}{\deg(X_q)}=\frac{(q^m+1)(q-1)}{q^{k-1}-1},\text{ that is } D\geq\left\lceil\frac{(q^m+1)(q-1)}{q^{k-1}-1}\right\rceil.\]

As a consequence, if $k\ll m$, the two Hilbert functions $\HF_S$ and $\HF_{S_{\Fm}}$ coincide for many values. We can then simply focus on $\HF_S(d)$.
\begin{proposition}\label{prop:hilbpolyofS}
  The Hilbert polynomial of $S$ is \[\HP_S(d)=\frac{q^{k-1}-1}{q-1}d+1-\delta(X_q,P_\infty),\] where $\delta(X_q,P_\infty)$ is the delta invariant of singularity of $X_q$ at the point $P_\infty$
  (see for instance \cite[Section~IV.1.2]{S59} for a definition).
\end{proposition}
\begin{proof}
    Since $X_q$ is a projective curve of degree $\frac{q^{k-1}-1}{q-1}$, its Hilbert polynomial $\HP_S(d)$ has degree one \cite[Theorem~7.5]{H77}, with leading coefficient $\frac{q^{k-1}-1}{q-1}$. Moreover, its constant term is equal to $1-g_a(X_q)$, where $g_a(X_q)$ is the arithmetic genus of $X_q$. By~\cite[Section~3]{elias2022sumsets}, \[g_a(X_q)=\sum_{P\in\mathrm{Sing}(X_q)}\delta(X_q,P)\,,\] where $\mathrm{Sing}(X_q)$ is the set of singular points of $X_q$ and $\delta(X_q,P)$ is the $\delta$-invariant of singularity of $X_q$ at the point $P$. Since $X_q$ is a monomial curve, its only two potential singular points are $P_0\coloneqq[1:0:\cdots:0]$ and $P_\infty\coloneqq[0:\cdots:0:1]$ (see~\cite[Section 3]{elias2022sumsets}). Let $\Gamma(\albf)$ and $\Gamma(\bebf)$ be the semigroups generated by $\albf=(\alpha_1,\dots,\alpha_k)$ and $\bebf=(\beta_1,\dots,\beta_k)$ respectively. Then, \[\delta(X_q,P_0)=\lvert\N\setminus \Gamma(\albf)\rvert \ \text{ and } \ \delta(X_q,P_\infty)=\lvert\N\setminus \Gamma(\bebf)\rvert\,.\] Since $\alpha_2=1$, $\Gamma(\albf)=\N$ and $P_0$ is a smooth point. In conclusion, $g_a(X_q)=\delta(X_q,P_\infty)$.
\end{proof}
We remark that obtaining a compact closed formula for $\delta(X_q,P_\infty)$ is, in general, not straightforward. In practice, one must proceed case by case and choose the most suitable strategy for its computation.

\subsection{Sumsets and Hilbert Regularity}
The goal of this section is to compute the Hilbert regularity of the coordinate ring $S$. To do so, we state the problem in terms of sumsets. We refer to~\cite[Chapter 1]{nathanson1996additive} for a general introduction to the topic. 

As in the previous subsection, we denote by $\Sa$ the non-ordered set of entries $\alpha_1,\dots,\alpha_k$ of the vector $\albf=\left(0,1,\frac{q^2-1}{q-1},\cdots,\frac{q^{k-1}-1}{q-1}\right)$. For every non-negative integer $s$, define the $s$-fold sumset of $\Gamma(\albf)$ as $0\Sa\coloneqq\{0\}$ and \[s\Sa\coloneqq\{\alpha_{i_1}+\dots+\alpha_{i_s}:1\leq i_1\leq\dots\leq i_s\leq k\}\] for all $s\geq1$. Then, there exists a direct relation among the $s$-fold sumsets of $\Gamma(\albf)$ and the Hilbert function of $S$, that is \[\HF_S(s)=\lvert\, s\Sa\,\rvert,\] see~\cite[Section 2]{eliahou2022iterated} or~\cite[Proposition 2.3]{elias2022sumsets}. In addition, the Castelnuovo--Mumford regularity of $S$ can as-well be expressed in terms of sumsets. 
\begin{theorem}[{\cite[Remark 3.6 and Theorem 3.7]{gimenez2023castelnuovo}}]\label{theorem:gimenez}
    The Castelnuovo--Mumford regularity of $S$ is \[\CMR(S)=\max\big\{\mathrm{AP}(\Gamma(\albf)), \,\mathrm{E}(\Gamma(\albf))\big\},\]
    where
    \begin{equation*}
        \begin{split}
            \mathrm{AP}(\Gamma(\albf))&\coloneqq\max\{s\in\N:\exists z\in s\Sa\text{ such that }z\notin (s-1)\Sa\text{ and }z-\alpha_k\notin (s-1)\Sa\}\\
            \mathrm{E}(\Gamma(\albf))&\coloneqq\max\{s\in\N:\exists z\in s\Sa\text{ such that }z-\alpha_k\in s\Sa\setminus (s-1)\Sa\},
        \end{split}
    \end{equation*}
    and $\mathrm{E}(\Gamma(\albf))\coloneqq-\infty$ if the above set on the right hand side is empty.
\end{theorem}
In light of this result, our goal is then to estimate $\mathrm{AP}(\Gamma(\albf))$ and $\mathrm{E}(\Gamma(\albf))$.

\begin{example}
    Before proceeding, we present a simple example to illustrate the above definitions and provide some intuition for the results that follow. Let $q=2$, $k=5$. Then $\albf=(0,1,3,7,15)$ and, by commutativity, for all $s\geq1$, \[s\Sa=\{\nu_1\cdot0+\nu_2\cdot1+\nu_3\cdot3+\nu_4\cdot7+\nu_5\cdot15\st \nu_1+\dots+\nu_5=s\}.\] In general, an integer may admit more than one representation as a linear combination of $\alpha_1,\dots,\alpha_k$. In particular, observe that whenever $z\in\N$ is such that $z\in s\Sa$ for some $s>0$, then $z=z+0=z+1\cdot\alpha_1\in (s+1)\Sa$. Since to determine $\mathrm{AP}(\Sa)$ and $\mathrm{E}(\Sa)$ we require that $z\notin(s-1)\Sa$, we can restrict to the representations having $\nu_1=0$. At the same time, to guarantee that $z\notin(s-1)\Sa$, we have to determine the representation of $z$ having minimum sum of coefficients. Let us consider all the possible representations of the number $z=8$ with $\nu_1=0$:
    \begin{align*}
        8 &= 8\cdot1 = 8\cdot\alpha_2\in 8\Sa \\
        &= 1\cdot 3+5\cdot1 = 1\cdot\alpha_3 + 5\cdot\alpha_2\in 6\Sa \\
        &= 2\cdot 3+2\cdot1 = 2\cdot\alpha_3 + 2\cdot\alpha_2\in 4S\Sa \\
        &= 1\cdot 7+1\cdot1 = 1\cdot\alpha_4 + 1\cdot\alpha_2\in 2\Sa
    \end{align*}
    In conclusion, since $8-15\notin\Gamma(\albf)$, we have that
    \begin{gather*}
        \max\{s\in\N:8\notin (s-1)\Sa\text{ and }8-\alpha_k\notin (s-1)\Sa\}=2 \\
        \{s\in\N:8-\alpha_k\in s\Sa\setminus (s-1)\Sa\}=\emptyset.
    \end{gather*}

    Let us now consider $z=20$. Then, one can verify that the representation of $z$ with minimum sum of coefficients is $20=1\cdot\alpha_5+1\cdot\alpha_3+2\cdot\alpha_2$, implying that $\min\{s\in\N\st 20\in s\Sa\}=4$. However, $20-15=5$ and its minimum representation is $5=1\cdot\alpha_3+2\cdot\alpha_1$, meaning that $5\in s\Sa$ for all $s\geq3$. Therefore, $z=20$ does not contribute in determining $\mathrm{AP}(\Gamma(\albf))$ and $\mathrm{E}(\Gamma(\albf))$.
\end{example}

\begin{proposition}\label{prop:reduced representation}
    Let $z,\nu_2,\dots,\nu_k$ be positive integers such that $z=\nu_2\alpha_2+\cdots+\nu_k\alpha_k$. Then, there exist some positive integers $\bar\nu_2,\dots,\bar\nu_k$ with $\bar\nu_j\leq q$ for all $j\in\{2,\dots,k-1\}$, and with the property that $\bar\nu_i=q$ implies $\bar\nu_j=0$ for all $j\in\{2,\dots,i-1\}$, such that $z=\bar\nu_2\alpha_2+\cdots+\bar\nu_k\alpha_k$ and $\sum_{i=2}^k\bar\nu_i\leq\sum_{i=2}^k\nu_i$.
\end{proposition}
\begin{proof}
  We order $\mathbb{N}^{k-1}$ with the graded reverse lexicographic order:
  $
    (\mu_2, \dots, \mu_k) \prec_{\text{grevlex}} (\eta_2, \dots, \eta_k)$
    if either $\sum_{i=2}^k \mu_i < \sum_{i=2}^k \eta_i$ or
    $\sum_{i=2}^k \mu_i = \sum_{i=2}^k \eta_i$
    and 
    $\mu_s > \eta_s$ for the largest index $s$ such that $\mu_s \neq \eta_s$.
  
    Let $z=\nu_2\alpha_2+\cdots+\nu_k\alpha_k$. First of all, suppose that $\nu_2\geq q+1$. Since $(q+1)\alpha_2=\alpha_3$, it is immediate to see that the coefficients $\nu_2,\nu_3$ can be substituted with $\nu_2-(q+1), \nu_3+1$, whose sum is smaller than $\nu_2+\nu_3$. Suppose now that there exist two indices $2\leq i <j\leq k-1$ such that $\nu_i>0$ and $\nu_j\geq q$. Since
    \begin{equation*}
        q\alpha_j+\alpha_i=q\frac{q^{j-1}-1}{q-1}+\frac{q^{i-1}-1}{q-1}=\frac{q^{j}-1}{q-1}+q\frac{q^{i-2}-1}{q-1}=\alpha_{j+1}+q\alpha_{i-1},
    \end{equation*}
    we can further decrease the sequence $(\nu_2,\dots,\nu_k)$ with respect to $\prec_{\text{grevlex}}$ into
    \[(\nu_2,\dots,\nu_{i-1}+q,\nu_i-1,\dots, \nu_j-q,\nu_{j+1}+1,\dots,\nu_k).\]
    Similarly, if there exists an index $j\in\{3,\dots,k-1\}$ such that $\nu_j\geq q+1$, since \[(q+1)\alpha_j=(q+1)\frac{q^{j-1}-1}{q-1}=\frac{q^j-1}{q-1}+\frac{q^{j-1}-1}{q-1}-1=\frac{q^j-1}{q-1}+q\frac{q^{j-2}-1}{q-1}=\alpha_{j+1}+q\alpha_{j-1},\] we can move from $\nu_2,\dots,\nu_{k}$ to $\nu_2,\dots, \nu_{j-1}+q,\nu_j-(q+1),\nu_{j+1}+1,\dots,\nu_k$, and achieve again a reduction with respect to
    $\prec_{\text{grevlex}}$.
    By iterating these transformations, the sequence of tuples $(\nu_2, \dots, \nu_k)$
    strictly decreases with respect to $\prec_{\text{grevlex}}$. Therefore, the process terminates and we obtain $\bar\nu_2,\dots,\bar\nu_k$ satisfying $\bar\nu_i\leq q$ for all $j\in\{3,\dots,k-1\}$ and that $\bar\nu_i=q$ implies $\bar\nu_j=0$ for $2\leq j<i$, with the additional property that $\sum_{i=2}^k\bar\nu_i\leq\sum_{i=2}^k\nu_i$.
\end{proof}

\begin{lemma}\label{lemma:AP}
    It holds that $\mathrm{AP}(\Gamma(\albf))= (q-1)(k-2)+1$.
\end{lemma}
\begin{proof}
  We begin by proving that $\mathrm{AP}(\Gamma(\albf))\leq (q-1)(k-2)+1$. Let $z=\alpha_{i_1}+\dots+\alpha_{i_{v+1}}\in(v+1)\Sa$ and assume that $z\notin v\Sa$. For $j\in\{1,\dots,k\}$, denote by $\nu_j$ the number of occurrences of $\alpha_j$ in the list $\alpha_{i_1},\dots,\alpha_{i_{v+1}}$. If $\nu_1>0$, then $z=z-\alpha_1\in v\Sa$. Hence, it follows that $\nu_1=0$. Moreover, by Proposition~\ref{prop:reduced representation}, we can assume $\nu_2,\dots,\nu_k$ to be such that $\nu_i\leq q$ for all $i\in\{2,\dots,k-1\}$, with the additional property that $\nu_i=q$ implies $\nu_j=0$ for $j\in\{2,\dots,i-1\}$. Finally, if $z-\alpha_k\notin v\Sa$, we can further restrict ourselves to the case $\nu_k=0$. Therefore, we conclude that
  \[ \mathrm{AP}(\Gamma(\albf))\leq v+1\leq q+(q-1)(k-3) = (q-1)(k-2)+1.\]
  On the other side, to prove that $\mathrm{AP}(\Gamma(\albf))\geq (q-1)(k-2)+1$, we only need to show that there exists $z\in v\Sa$ such that $z\notin (v-1)\Sa$ and $z-\alpha_k\notin (v-1)\Sa$, with $v=(q-1)(k-2)+1$. Consider $z=q\alpha_2+(q-1)\alpha_3+\dots+(q-1)\alpha_{k-1}$. It follows by definition that $z\in v\Sa$. Moreover, using the fact that \[q\alpha_j=q\frac{q^{j-1}-1}{q-1}=\frac{q^{j}-1}{q-1}+1=\alpha_{j+1}-1\] for all $j\in\{1,\dots,k-1\}$, we obtain
    \begin{align*}
        z &=q\alpha_2+(q-1)\alpha_3+\dots+(q-1)\alpha_{k-1} \\
        &=(\alpha_3-1)+(\alpha_4-1-\alpha_3)+\dots+(\alpha_k-1-\alpha_{k-1}) \\
        &=\alpha_k-(k-2)<
        \alpha_k.
    \end{align*} Hence, $z-\alpha_k\notin(v-1)\Sa$. It remains to show that $z\notin(v-1)\Sa$. If this was the case, then it would be possible to reduce the sum of the coefficients in the representation of $z$ above. Let $ \nu_2\alpha_2+\dots+ \nu_{k}\alpha_{k}$ be such an alternative representation of $z$ with $\sum_{i=2}^{k}\nu_i=v-1$. By Proposition~\ref{prop:reduced representation}, we can force $\nu_2\leq q$ and $\nu_j\leq q-1$ for all $j\in\{3,\dots,k-1\}$. Moreover, since $z<\alpha_k$, we must have $\nu_k=0$. Since we are assuming that $z\in(v-1)\Sa$, either $\nu_2<q$ or there exists an index $j\in\{2,\dots,k-1\}$ for which $\nu_j<q-1$, and this leads to a contradiction.
\end{proof}
\begin{lemma}\label{lemma:E}
    We have that $\mathrm{E}(\Gamma(\albf))=-\infty$.
\end{lemma}
\begin{proof}
  Suppose there exist $s,z \in \mathbb{N}$ such that
  $z=\alpha_1\nu_1+\dots+\alpha_k\nu_k\in s\Sa$ and $z-\alpha_k\in s\Sa\setminus (s-1)\Sa$. By Proposition~\ref{prop:reduced representation}, we can assume that there exists an index $i\geq 2$ such that $\nu_i\leq q$, $\nu_{h}=0$ for all $h\in\{1,\dots,i-1\}$, and $\nu_j\leq q-1$ for all $j\in\{i+1,\dots,k-1\}$. At the same time, since $z-\alpha_k\in s\Sa\setminus (s-1)\Sa$, we get $z\geq \alpha_k$. Since \[\nu_2\alpha_2+\dots+\nu_{k-1}\alpha_{k-1}\leq\max_{i\in\{2,\dots,k-1\}}\{q\alpha_i+(q-1)\alpha_{i+1}+\dots+(q-1)\alpha_{k-1}\}<\alpha_k,\] it then follows that $\nu_k\geq 1$. Therefore, $z-\alpha_k\in (s-1)\Sa$, that is in contradiction with our assumption. We conclude that $$\{s\in\N:\exists z\in s\Sa\text{ such that }z-\alpha_k\in s\Sa\setminus (s-1)\Sa\}=\emptyset,$$ and therefore, $\mathrm{E}(\Gamma(\albf))=-\infty$.
\end{proof}

Combining Theorem~\ref{theorem:gimenez} with Lemma~\ref{lemma:AP} and Lemma~\ref{lemma:E}, we finally conclude that $$\CMR(S)=(q-1)(k-2)+1.$$ In the next theorem we show that $S$ is Cohen-Macaulay and, as a consequence, we determine the value of the Hilbert regularity of $S$.
\begin{theorem}\label{theorem:CMregMonomialcurve}
    With the notation above, the Hilbert regularity of the homogeneous coordinate ring $S$ of $X_q$ is \[\HR(S)=(q-1)(k-2).\]
\end{theorem}
\begin{proof}
    Let $\Gamma(\albf,\bebf)$ be the semigroup generated by the pairs $(\alpha_1,\beta_1),\dots,(\alpha_k,\beta_k)$. Consider the group
    $$G\coloneqq\left\{(n_1,n_2)\in\Z^2:n_1+n_2\equiv0\mod{\frac{q^{k-1}-1}{q-1}}\right\}$$ and the semigroup
    \begin{equation*}
        \Gamma'(\albf,\bebf)\coloneqq \left\{(n_1,n_2)\in G:(n_1,n_2)+\left(\frac{q^{k-1}-1}{q-1},0\right),(n_1,n_2)+\left(0,\frac{q^{k-1}-1}{q-1}\right)\in \Gamma(\albf,\bebf)\right\}.
    \end{equation*}
    In~\cite[Theorem 2.6]{goto1976affine} it was proved that $S$ is Cohen-Macaulay if and only if $\Gamma'(\albf,\bebf)\setminus \Gamma(\albf,\bebf)=\emptyset$, see also~\cite[Theorem 1.1]{reid2005non}.
    Let $(n_1,n_2)\in \Gamma'(\albf,\bebf)$. By definition
    \begin{equation*}
        \left(n_1+\frac{q^{k-1}-1}{q-1},n_2\right)=\sum_{i=1}^{k} \nu_i(\alpha_i,\beta_i)\in \Gamma(\albf,\bebf),
    \end{equation*} for some $\nu_i\geq0$. By Proposition~\ref{prop:reduced representation}, there exist $\bar \nu_2,\dots, \bar \nu_k$ such that 
    $$n_1+\frac{q^{k-1}-1}{q-1}=\bar \nu_2\alpha_2+\dots+\bar \nu_k\alpha_k,$$ $\sum_{i\geq2} \nu_i\geq\sum_{i\geq2}\bar \nu_i$, $\bar \nu_i\leq q$ for all $i\in\{2,\dots,k-1\}$, and such that $\bar \nu_i=q$ implies $\bar \nu_j=0$ for all $j\in\{2,\dots,i-1\}$. Then we have
    \begin{equation*}
        n_2=\left(\sum_{i=1}^{k}\nu_i\right)\frac{q^{k-1}-1}{q-1}-\sum_{i=2}^{k}\nu_i \alpha_i= \left(\sum_{i=1}^{k}\nu_i\right)\frac{q^{k-1}-1}{q-1}-\sum_{i=2}^{k}\bar \nu_i \alpha_i=\sum_{i=1}^k\bar\nu_i\beta_i,
    \end{equation*}
    where $\bar\nu_1=\sum_{i=1}^{k}\nu_i-\sum_{i=2}^{k}\bar\nu_i$.
    As a consequence, we get
    \begin{equation} \label{eq: elt repr}
        \left(n_1+\frac{q^{k-1}-1}{q-1},n_2\right)=\sum_{i=1}^{k} \bar\nu_i(\alpha_i,\beta_i)\in \Gamma(\albf,\bebf).
    \end{equation}
    If $\bar\nu_k=0$, we obtain from~\eqref{eq: elt repr} that
    \begin{align} \label{eq: semigroups}
        n_1+\frac{q^{k-1}-1}{q-1}&\leq\max_{i\in\{2,\dots,k-1\}} \left\{q\,\alpha_i+\sum_{j=i+1}^{k-1}(q-1)\alpha_j\right\} \nonumber \\
        &=\max_{i\in\{2,\dots,k-1\}} \left\{q\frac{q^{i-1}-1}{q-1}+\sum_{j=i+1}^{k-1}(q-1)\frac{q^{j-1}-1}{q-1}\right\} \nonumber \\
        &=\max_{i\in\{2,\dots,k-1\}} \left\{(q+\dots+q^{i-1})+\big(q^i+\dots+q^{k-2}-(k-1-i)\big)\right\} \nonumber \\
        &=\max_{i\in\{2,\dots,k-1\}} \left\{\frac{q^{k-1}-1}{q-1}-k+i\right\} = \frac{q^{k-1}-1}{q-1}-1.
    \end{align}
    Since $(n_1,n_2)\in \Gamma'(\albf,\bebf)$ implies $n_1+\frac{q^{k-1}-1}{q-1}\in\Gamma(\albf)$ and all the elements in $\albf$ are non-negative, it must be that $n_1\geq0$, which contradicts~\eqref{eq: semigroups}. Hence, $\bar\nu_k>0$ and we immediately derive from~\eqref{eq: elt repr} that \[(n_1,n_2)=\sum_{i=1}^{k-1} \bar\nu_i(\alpha_i,\beta_i)+(\bar\nu_k-1)(\alpha_k,\beta_k)\in \Gamma(\albf,\bebf).\]
    Therefore, we conclude that $\Gamma(\albf,\bebf)=\Gamma'(\albf,\bebf)$, and $S$ is Cohen-Macaulay. In particular, the depth of $S$ is equal to $2$ and by~\cite[Corollary 4.8]{eisenbud2005geometry} we conclude that
    \begin{equation*}
        \HR(S)=1-\mathrm{depth}(S)+\CMR(S)=(q-1)(k-2). \qedhere
    \end{equation*}
\end{proof} \smallskip

It remains to determine the Hilbert regularity of the homogeneous coordinate ring $S_{\Fm}$. To this end, we develop an analogue of Theorem~\ref{theorem:gimenez} in the context of finite abelian groups. \\ Denote by $\Z/(q^m-1)\Z$ the cyclic group of order $q^m-1$. We set $$T_1\coloneqq\left\{1, \frac{q^2-1}{q-1},\dots,\frac{q^{k-2}-1}{q-1}\right\}\subset\Z/(q^m-1)\Z.$$ For all $d\geq2$, we recursively define the set $T_d\subseteq\Z/(q^m-1)\Z$ as \[T_d\coloneqq T_{d-1}\cup (T_{d-1}+1)\cup \left(T_{d-1}+\frac{q^2-1}{q-1}\right)\cup\dots\cup\left(T_{d-1}+\frac{q^{k-1}-1}{q-1}\right).\] Note that $T_d$ does not equal $d\Sa$ modulo $q^m-1$ since $T_1$ is obtained from $\Gamma(\albf)$ by removing $0$ and $\frac{q^{k-1}-1}{q-1}$.
\begin{proposition}\label{prop:regTd}
    For all $d\geq1$, it holds that
    \begin{equation*}
        \HF_{S_{\Fm}}(d)=\lvert T_d\rvert+2,
    \end{equation*}
    and
    \begin{equation*}
        \HR(S_{\Fm})=\min\left\{d:\lvert T_d\rvert=q^{m}-1\right\}.
    \end{equation*}
\end{proposition}
\begin{proof}
    Recall that, from \cite[Chapter~9, Proposition~9]{CLO25}, for any degree compatible monomial order $\leq$, the Hilbert regularity of $S_{\Fm}$ is equal to the Hilbert regularity of $\Fm[x_1,\dots,x_k]/\mathrm{in}_{\leq}(I(X(\Fm)))$, where $\mathrm{in}_{\leq}(I(X(\Fm)))$ is the initial ideal of $I(X(\Fm))$. By Lemma~\ref{lemma:fbasis}, $\mathrm{in}_{\leq}(I(X_q(\Fm)))$ is generated by all monomials $x^{\gabf}$ of degree $d$ for which there exists a monomial $x^{\bm{\delta}}$ of degree $d$ such that $\gabf\geq\bm{\delta}$, $\langle\gabf,\albf\rangle\equiv\langle\bm{\delta},\albf\rangle\pmod{q^{m}-1}$, and such that $x^{\gabf}-x^{\bm{\delta}}$ vanishes on $P_0$ and $P_{\infty}$.  Therefore, to compute $\mathrm{HF}_{S_{\Fm}}(d)$, we need to determine the number of monomials of degree $d$ for which one of these properties does not hold. For any $d\geq 1$ consider the set
    \begin{equation*}
        Z_d\coloneqq\{\gabf\in\Z_{\geq0}^k:\lvert\gabf\rvert=d\}\setminus\{(d,0,\dots,0),(0,\dots,0,d)\}.
    \end{equation*}
    and the function $f_d:Z_d\longrightarrow \Z/(q^m-1)\Z,\gabf\longmapsto\langle\gabf,\albf\rangle$. Note that for every element $a$ in the image of $f_d$ there exists only one pre-image $\gabf_a$ in $Z_d$ that is minimal according to the order that has been fixed. Then,
    \begin{equation*}
        \{x^{\gabf_a}:a\in\mathrm{Im}(f_d)\}\cup\{x_1^d,x_k^d\}
    \end{equation*}
    is the set of monomials of degree $d$ that generate $(S_{\Fm})_d$. Since $$\lvert \{x^{\gabf_a}:a\in\mathrm{Im}(f_d)\}\rvert=\lvert T_d\rvert,$$ we obtain $\HF_{S_{\Fm}}(d)=\lvert T_d\rvert+2$. \\ In conclusion, note that $d\in T_d$ for any positive integer $d$, implying $T_d=\Z/(q^m-1)\Z$ for $d$ large enough. Recalling that $\lvert X_q(\Fm)\rvert=q^m+1$, we finally get 
    \begin{equation*}
        \HR(S_{\Fm})=\min\left\{d:\lvert T_d\rvert=q^{m}-1\right\}.\qedhere
    \end{equation*}
\end{proof}
\begin{proposition}\label{prop:hreg0dimension}
    Let $q,k,m$ be such that 
    \begin{equation}\label{eq:hypothesis}
         \left\lfloor\frac{(q^m-1)(q-1)}{q^{k-1}-1}\right\rfloor\geq(q-1)(k-2).
     \end{equation}
     Then, the Hilbert regularity of the homogeneous coordinate ring $S_{\Fm}$ of $X_q(\Fm)$ is \[\HR(S_{\Fm})=\left\lfloor\frac{(q^m-1)(q-1)}{q^{k-1}-1}\right\rfloor+(q-1)(k-2)+\delta,\]
     where $\delta$ can be either $0$ or $1$.
\end{proposition}
\begin{proof}
    By Proposition~\ref{prop:reduced representation} there exist non-negative integers $\nu_1,\dots,\nu_k$ with minimal sum $\Sigma_z:=\nu_1+\dots+\nu_k$ such that $z=\nu_1\alpha_1+\cdots+\nu_k\alpha_k$, $\nu_j\leq q$ for all $j\in\{2,\dots,k-1\}$, and such that $\nu_i=q$ implies $\nu_j=0$ for $j<i$. Since
    \begin{align*}
        \nu_2\alpha_2&+\cdots+\nu_{k-1}\alpha_{k-1}\leq \max_{2\leq i\leq k-1} q\alpha_i+(q-1)(\alpha_{i+1}+\dots+\alpha_{k-2})\\
        &=\max_{2\leq i\leq k-1} q+\dots+q^{i-1}+(q^{i}-1)+\dots+(q^{k-2}-1)< \frac{q^{k-1}-1}{q-1}=\alpha_k,
    \end{align*}
    we get $\nu_k=\left\lfloor\frac{z}{\alpha_k}\right\rfloor$. Moreover, by~\eqref{eq:hypothesis}, we obtain that $\sum_z<\sum_{z+l(q^{m}-1)}$ for every positive integer $l$, see also Remark~\ref{remark:l>0}. As a consequence,  by Proposition~\ref{prop:regTd}, $\HR(S_{\Fm})$ can be expressed as
    \begin{equation}\label{eq:Xdsumz}
        \HR(S_{\Fm})=\min\left\{d:\lvert T_d\rvert=q^{m}-1\right\}=\max\{\Sigma_z\st z\leq {q^m-1}\}.
    \end{equation}
    Since the largest number that we want to represent is $q^m-1$, it follows that $\nu_k$ is going to be at most $\left\lfloor\frac{q^m-1}{\alpha_k}\right\rfloor$. Let $\bar z$ be the integer in $\left\{\left\lfloor\frac{q^m-1}{\alpha_k}\right\rfloor\alpha_k,\dots,q^m-1\right\}$ with maximum sum of the coefficients, and let $r\in\{0,\dots,\alpha_k-1\}$ be such that $\bar z=\left\lfloor\frac{q^m-1}{\alpha_k}\right\rfloor\alpha_k+r$. In particular, as we have seen above, it holds that $\Sigma_r\in\{0,\dots,(q-1)(k-2)+1\}$, which implies
    \begin{gather*}
        \max\{\Sigma_z\st z\in\Z/(q^m-1)\Z\} \\ \verteq \\[-10pt] \max\left\{\left\{\Sigma_{\bar z} \st\left\lfloor\frac{q^m-1}{\alpha_k}\right\rfloor\alpha_k\leq\bar z \leq q^m-1\right\},\left\{\Sigma_{\bar z} \st 0\leq\bar z<\left\lfloor\frac{q^m-1}{\alpha_k}\right\rfloor\alpha_k\right\}\right\} \\[-4pt] \verteq \\[-4pt] \max\left\{\left\lfloor\frac{q^m-1}{\alpha_k}\right\rfloor+\Sigma_r\,,\left\lfloor\frac{q^m-1}{\alpha_k}\right\rfloor-1+(q-1)(k-2)+1\right\} \\ \verteq \\[-10pt] \left\lfloor\frac{q^m-1}{\alpha_k}\right\rfloor+(q-1)(k-2)+\delta,
    \end{gather*}
    where $\delta$ is either $0$ or $1$. Hence, we conclude by~\eqref{eq:Xdsumz}.
\end{proof}
\begin{remark}\label{remark:l>0}
    Condition~\eqref{eq:hypothesis} in the statement of Proposition~\ref{prop:hreg0dimension} ensures that $\sum_{z+l(q^{m}-1)}>\sum_z$ for every $l>0$ and $z\in\N$. In general, this may not hold if \begin{equation*}
         \left\lfloor\frac{(q^m-1)(q-1)}{q^{k-1}-1}\right\rfloor\ll(q-1)(k-2).
     \end{equation*}
     In particular, for certain parameter choices, there may exist $l$ such that $\sum_{\bar z+l(q^{m}-1)}<\sum_{\bar z}$, where $\bar z\in\N$ achieves the maximum in~\eqref{eq:Xdsumz}. If this is the case, then the value of the Hilbert regularity should drop further. However, identifying the exact conditions under which this occurs is not straightforward. Since this behavior happens only rarely, we assumed Condition~\eqref{eq:hypothesis} in order to simplify both the statement and the proof of the proposition.
\end{remark}
\begin{remark}
    As mentioned above, we remark that the results of this section can be further generalized to the case $j>1$. More precisely, one may analogously study the $(q^m/q)$-Hilbert function of the homogeneous coordinate rings \[S\coloneqq\Fm[x_1,\dots,x_k]/I(X_{q^j}) \ \text{ and } \ S_{\Fm}\coloneqq\Fm[x_1,\dots,x_k]/I(X_{q^j}(\Fm)).\] To this end, one first needs to adapt the definition of the matrix $A$ in~\eqref{eq: matrix A} that allows to see $I(X_{q^j})$ as the span of an $\Fm$-vector space. Once this is done, the following results can be generalized accordingly.
\end{remark}

\subsection{Linearized Reed--Solomon Codes with an Arbitrary Number of Blocks} \label{LRS codes with an arbitrary number of blocks}
The analysis carried out in the previous subsection can be applied analogously to study the $(q^m/q)$-Hilbert sequence of a (doubly-extended) linearized Reed--Solomon code $\mC$ with $t<q+1$ blocks of maximal length. In this case, one considers the Hilbert sequence of the homogeneous coordinate ring associated with the set of projective points $\bigsqcup_{i=1}^t \mL_{\mU_i}\subset X_q(\Fm)$ and adapts the final results of the previous subsection to determine the corresponding Hilbert regularity. In particular, if $\mC$ is a standard linearized Reed--Solomon code, by Theorem~\ref{theorem: linear set of lrs}, for every $i \in \{1,\dots,t\}$ we have $\mL_{\mU_i} = \mV_{a_i}(\Fm)$ for some $a_i \in \Fm^*$. Moreover, by Proposition~\ref{prop: H seq and CM reg}, we already know that the final stable value reached by the $(q^m/q)$-Hilbert sequence of $\mC$ is $\lvert \, \bigsqcup_{i=1}^t \mV_{a_i}(\Fm)\, \rvert = t\frac{q^m-1}{q-1}$. The doubly-extended case is analogous. The only difference is that $\mL_{\mU_{t-1}} = P_0$ and $\mL_{\mU_t} = P_\infty$, so that the final stable dimension becomes $(t-2)\frac{q^m-1}{q-1}+2$.

By further restricting to a single block, the same strategies presented here can be adapted to study the $(q^m/q)$-Hilbert sequence of a $k$-dimensional (generalized) Gabidulin code $\mC$ over $\Fm^m$. In particular, up to a change of basis, one can assume that $\mC$ is generated by a matrix    
\[G=\left(\begin{array}{cccc}
    1 & & & \\
    & N_1(a) & & \\
    &  & \ddots & \\
    & & & N_{k-1}(a) \end{array}\right) \ \left(\begin{array}{ccc}
    b_1 & \cdots & b_m\\
    \sigma(b_1) & \cdots & \sigma(b_m) \\
    \vdots & & \vdots \\
    \sigma^{k-1}(b_1) & \cdots & \sigma^{k-1}(b_m) \\
\end{array}\right)\] of the form $\eqref{eq: gen matrix LRS}$ above, for some $a\in\Fm^*$ and some $\Fq$-linearly independent elements $b_1,\dots,b_m\in\Fm$. In this case, it follows by Theorem~\ref{theorem: linear set of lrs} that the geometric object associated with $\mC$ is $\mL_G=\mV_a(\Fm)$. Hence, studying the $(q^m/q)$-Hilbert sequence of $\mC$ reduces to studying the Hilbert sequence of $\mV_a(\Fm)$. In this way, the results presented here extend the analysis of the $(q^m/q)$-Hilbert sequence of Gabidulin codes carried out in~\cite{astore2025geometric}. Specifically, for the cases in which $k-1$ divides $m$, we can derive the following result.
\begin{theorem}
    Let $\mC$ be a $k$-dimensional Gabidulin code over $\Fm^m$, where $m$ is a positive multiple of $k-1$. Let also $a\in\Fm^*$ be such that $\mV_a(\Fm)$ is the geometric object associated with $\mC$. Then, the vanishing ideal of $\mV_a(\Fm)$ over $\Fm[x_1,\dots,x_k]$ is given by \[I(\mV_a(\Fm))=I(X_q)+\,(p_a)=\big(\{x_jx_{l-1}^q-x_{j-1}^qx_l\st1\leq j<l\leq k\}\big)+\,(p_a),\] where $(p_a)$ denotes the ideal generated by the polynomial \[p_a(x_1,\dots,x_k):=N_m(a)\,x_1^{\frac{q^m-1}{q^{k-1}-1}}-x_k^{\frac{q^m-1}{q^{k-1}-1}}\in\Fm[x_1,\dots,x_k].\]
\end{theorem}
\begin{proof}
    Let $P_{\bar a}=\left(1:\bar a:\dots:\bar a^{\frac{q^{k-1}-1}{q-1}}\right)\in X_q(\Fm)$. Then, by Lemma~\ref{lemma: different norms different sets}, $P_{\bar a}\in\mV_a(\Fm)$ if and only if $N_m(\bar a)=N_m(a)$. Equivalently, since \[N_m(\bar a)=\bar a^{\frac{q^m-1}{q-1}}=\bar a^{\frac{q^{k-1}-1}{q-1}\frac{q^m-1}{q^{k-1}-1}},\] $P_{\bar a}$ is a zero of $p_a(x_1,\dots,x_k)=N_m(a)\,x_1^{\frac{q^m-1}{q^{k-1}-1}}-x_k^{\frac{q^m-1}{q^{k-1}-1}}$, which is well-defined since $k-1$ divides $m$. Denote by $(p_a)$ the ideal generated by $p_a$ over $\Fm[x_1,\dots,x_k]$ and by $V(p_a)$ the variety defined by it. Then, $\mV_a(\Fm)\subseteq X_q(\Fm)\cap V(p_a)\subseteq X_q\cap V(p_a)$. Moreover, by Bézout's theorem, we have \[|\,X_q\cap V(p_a) \,| \leq \deg(X_q)\deg(p_a)= \frac{(q^{k-1}-1)(q^m-1)}{(q-1)(q^{k-1}-1)}=\frac{q^m-1}{q-1}=|\mV_a(\Fm)|.\] As a consequence, $\lvert X_q\cap V(p_a)\rvert=\lvert\mV_a(\Fm)\rvert$. Therefore, the intersection is reduced (\emph{i.e.}, every point has intersection multiplicity one) and is equal to $\mathcal{V}_a(\Fm)$. This implies that the two varieties intersect transversally, and $I(X_q)+(p_a)$ is the radical ideal defining $\mathcal{V}_a(\Fm)$.
\end{proof}
As a direct consequence of the previous theorem and Lemma~\ref{lemma: LRS partition normal curve}, we derive the following.
\begin{corollary} \label{cor: extra polynomials}
  Let $m$ be a positive multiple of $k-1$ and let $\mN^*$ be a set of $t\leq q-1$ elements in $\Fm^*$ with pairwise distinct norms. For all $a\in\mN^*$, let
  \[ p_a(x_1,\dots,x_k) \coloneqq N_m(a)\,x_1^{\frac{q^m-1}{q^{k-1}-1}}-x_k^{\frac{q^m-1}{q^{k-1}-1}}\in\Fm[x_1,\dots,x_k].\] Then, \[I\left(\bigsqcup_{a\in\mN^*}\mV_a(\Fm)\right)=\left(\prod_{a\in\mN^*}p_a\right)+I(X_q).\] In particular, if $t=q-1$, we obtain \[I\big(X_q(\Fm)\setminus\{P_0,P_\infty\}\big)=\left(\prod_{a\in\mN^*}p_a\right)+I(X_q).\]
\end{corollary}
To illustrate our theoretical findings, we performed some computational experiments to study and visualize the Hilbert sequences of linearized Reed--Solomon codes using {\sc Magma}~\cite{bosma1997magma} (see \cite{CiteGithub} for the detailed implementation). The results are presented in Figure~\ref{Figure: LRS} and Figure~\ref{Figure: comparison}.
\begin{example}
    Figure~\ref{Figure: LRS} displays the $(q^m/q)$-Hilbert sequence of a $[(4,4,4,4), 4]_{5^4/5}$ linearized Reed--Solomon code $\mC$, together with some notable points. In particular, the value $8=(q-1)(k-2)=\HR(S)$ marks the step at which the Hilbert sequence begins to coincide with the Hilbert polynomial $31x=\frac{q^{k-1}-1}{q-1}x$. The sequence continues to agree with the Hilbert polynomial up to step $20=\left\lfloor\frac{(q^m-1)(q-1)}{q^{k-1}-1}\right\rfloor=\HR(S)$, after which it slightly deviates before reaching its final stabilization. This behavior is a direct consequence of Corollary~\ref{cor: extra polynomials}, since from this degree onward the vanishing ideal of $\mL_{\mU_\mC}$ is no longer generated solely by the defining polynomials of $I(X_q)$. The sequence then proceeds until it attains its final stable dimension $624=q^m-1$ at step $28=\left\lfloor\frac{(q^m-1)(q-1)}{q^{k-1}-1}\right\rfloor+(q-1)(k-2)=\HR\left(S_{X_q(\Fm)\setminus\{P_0,P_\infty\}}\right)$. Note that this value is coherent with Proposition~\ref{prop: H seq and CM reg} as this example tests a standard linearized Reed--Solomon code with $q-1$ blocks whose associated geometric object is $X_q\setminus\{P_0,P_\infty\}$ by Theorem~\ref{theorem: linear set of lrs}.
    \begin{figure}[ht]
        \centering
        \includegraphics[width=\textwidth]{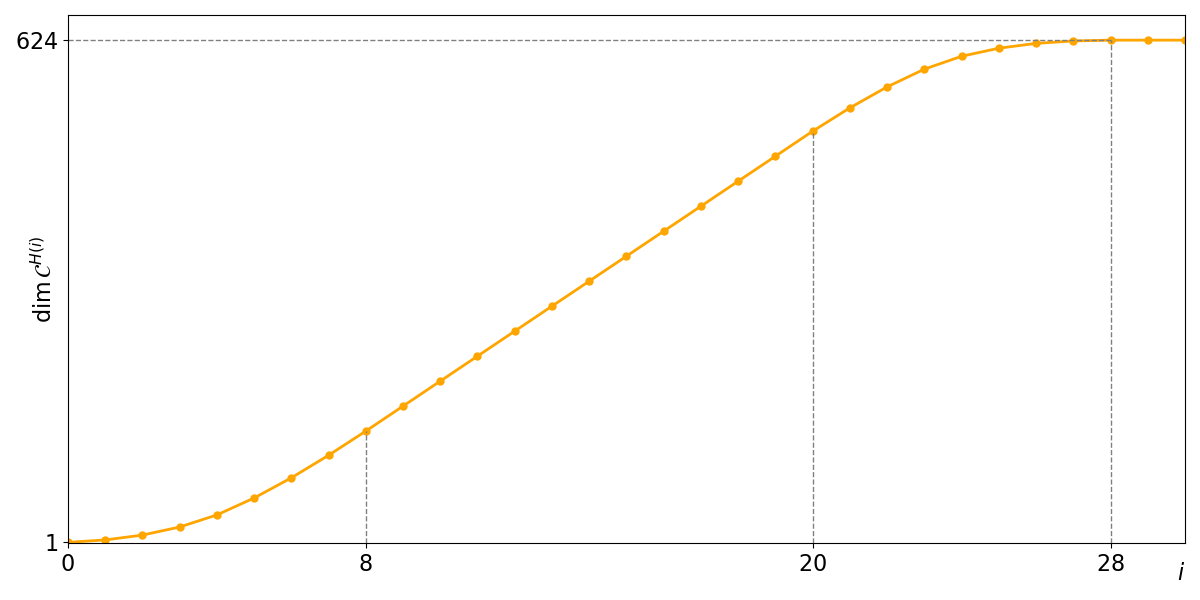}
        \caption{$(q^m/q)$-Hilbert sequence of a $[(4,4,4,4), 4]_{5^4/5}$ linearized Reed--Solomon code.} \label{Figure: LRS}
    \end{figure}
\end{example}

\begin{example}
    Figure~\ref{Figure: comparison} represents the $(q^m/q)$-Hilbert sequences of two $[(4,4,4), 4]_{4^4/4}$ codes. In particular, the orange line is related to a linearized Reed--Solomon code, $\mC$, while the blue line is associated to a code $\mG$ obtained as the concatenation of three unrelated Gabidulin codes. This example extends Remark~\ref{remark:concatenation not LRS} and confirms the expected fact that the different geometric objects associated with these two codes are also reflected in their $(q^m/q)$-Hilbert sequences. In particular, after initially coinciding, the two sequences diverge at step $5=q+1$, when the $(q^m/q)$-Hilbert sequence of $\mG$ begins to grow faster than that of $\mC$. As pointed out in Remark~\ref{remark:concatenation not LRS}, this behavior is related with the fact that the geometric objects associated with each block of $\mG$ do not necessarily lie on the same $q$-rational normal curve. Hence, the $(q+1)$-th part of $I(\mL_{\mU_\mG})$ will in general not contain all the polynomials determined in Proposition~\ref{prop: determinantal variety} to identify $X_q$, while $I(X_q)_{q+1}\subseteq I(\mL_{\mU_\mC})_{q+1}$. The two sequences then reach their final stable dimension $255=q^m-1$ at steps $13$ and $18=\left\lfloor\frac{(q^m-1)(q-1)}{q^{k-1}-1}\right\rfloor+(q-1)(k-2)=\HR\left(S_{X_q(\Fm)\setminus\{P_0,P_\infty\}}\right)$, respectively.
    \begin{figure}[ht]
        \centering
        \includegraphics[width=\textwidth]{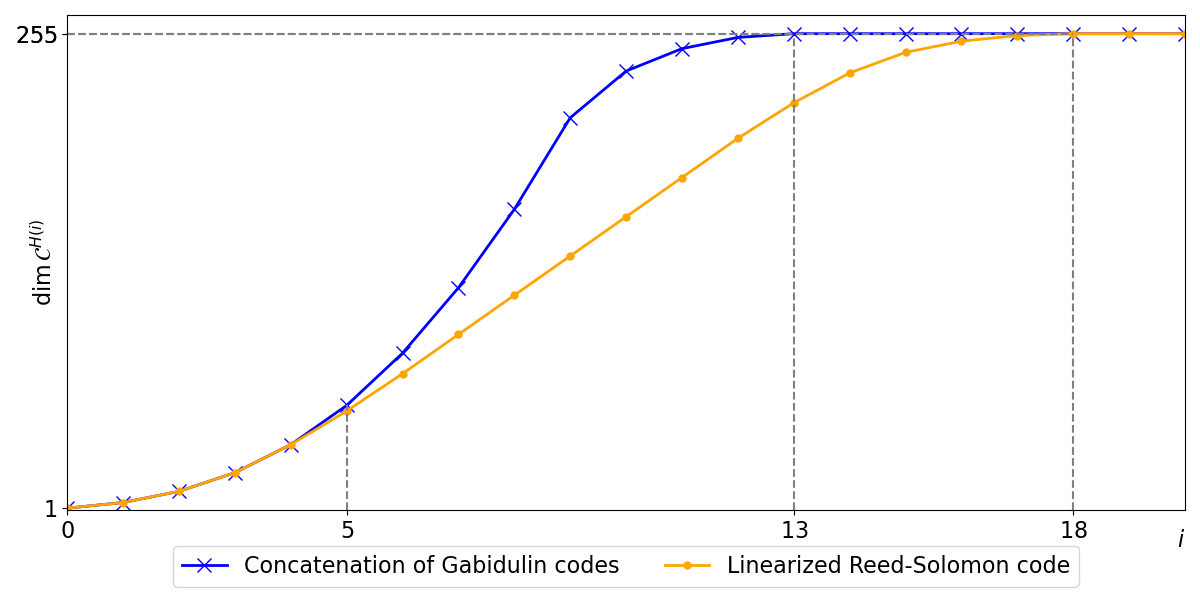}
        \caption{$(q^m/q)$-Hilbert sequences of $[(4,4,4), 4]_{4^4/4}$ codes.} \label{Figure: comparison}
    \end{figure}
\end{example}
\bigskip

\section*{Acknowledgments}
The authors of this paper would like to thank Philippe Gimenez, Mario González-Sánchez, Christophe Levrat, Alessandro Neri, and Ferdinando Zullo for insightful discussions and suggestions.

\bigskip

\bibliographystyle{abbrv}
\bibliography{biblio.bib}

\end{document}